\documentclass[11pt]{article}
\usepackage{amsmath,amssymb,amsfonts,amsthm,epsfig}
\usepackage{physics}
\usepackage[usenames,dvipsnames]{xcolor}

\usepackage{bm,xspace}

\usepackage{cancel}
\usepackage{fullpage}
\usepackage{liyang}
\usepackage{framed}
\usepackage{verbatim}
\usepackage[shortlabels]{enumitem}
\usepackage{array}
\usepackage{multirow}
\usepackage{afterpage}
\usepackage{mathrsfs}
\usepackage{caption}
\usepackage[normalem]{ulem}

\newcommand{\rnote}[1]{\footnote{{\bf \color{red}Rocco}: {#1}}}
\newcommand{\snote}[1]{\footnote{{\bf \color{blue}Sandip}: {#1}}}

\usepackage{caption} 

\usepackage{todonotes}

\usepackage[nameinlink]{cleveref}

\makeatletter
\newtheorem*{rep@theorem}{\rep@title}
\newcommand{\newreptheorem}[2]{
\newenvironment{rep#1}[1]{
 \def\rep@title{#2 \ref{##1}}
 \begin{rep@theorem}\itshape}
 {\end{rep@theorem}}}
\makeatother
\theoremstyle{plain}

\makeatletter
\newenvironment{proofof}[1]{\par
  \pushQED{\qed}%
  \normalfont \topsep6\p@\@plus6\p@\relax
  \trivlist
  \item[\hskip\labelsep
\emph{    Proof of #1\@addpunct{.}}]\ignorespaces
}{%
  \popQED\endtrivlist\@endpefalse
}
\makeatother

\newcommand{\ignore}[1]{}

\def\colorful{0}

\ifnum\colorful=1

\newcommand{\blue}[1]{{{\color{blue}#1}}}

\newcommand{\gray}[1]{{\color{gray}{#1}}}

\fi
\ifnum\colorful=0

\newcommand{\blue}[1]{{{#1}}}

\newcommand{\gray}[1]{{{#1}}}

\fi

\usepackage{boxedminipage}

\newreptheorem{theorem}{Theorem}
\newtheorem*{theorem*}{Theorem}
\newreptheorem{lemma}{Lemma}
\newreptheorem{proposition}{Proposition}
\newreptheorem{corollary}{Corollary}
\newtheorem*{noclaim*}{Claim}

\newcommand{\SW}{\mathrm{SW}}

\newcommand{\Del}{\mathrm{Del}}

\newcommand{\Indicator}[1]{\mathbf{1}\left[#1\right]}

\newcommand{\B}{\mathrm{B}}

\def\PAIR{\textnormal{PAIR}}
\def\SUBWORD{\textnormal{SUBWORD}}

\usepackage{dsfont}

\newcommand{\worst}{\mathrm{worst}}
\newcommand{\subword}{\mathrm{subword}}
\newcommand{\id}{\mathrm{id}}

\begin{document}

\title{Polynomial-time trace reconstruction\\ 
in the smoothed complexity model\vspace{0.2cm}}

\author{
Xi Chen\thanks{Supported by NSF grants CCF-1703925 and IIS-1838154.} 
\\
Columbia University\\
xichen@cs.columbia.edu
\and
Anindya De\thanks{Supported by NSF grants CCF-1926872 and CCF-1910534.}
\\
University of Pennsylvania\\%Penn\\
anindyad@cis.upenn.edu
\and
Chin Ho Lee\thanks{Supported by a grant from the Croucher Foundation and by the Simons Collaboration on Algorithms and Geometry.}
\\
Columbia University\\
c.h.lee@columbia.edu
\and
Rocco A. Servedio\thanks{Supported by NSF grants CCF-1814873, IIS-1838154, CCF-1563155, and by the Simons Collaboration on Algorithms and Geometry.} 
\\
Columbia University\\
rocco@cs.columbia.edu
\and
Sandip Sinha\thanks{Supported by NSF grants  CCF-1714818, CCF-1822809, IIS-1838154, CCF-1617955, CCF-1740833, and by the Simons Collaboration on Algorithms and Geometry.}
\\
Columbia University\\
sandip@cs.columbia.edu
}

\maketitle

\thispagestyle{empty}

%!TEX root = main.tex

\begin{abstract}
In the \emph{trace reconstruction problem}, an unknown source string $x \in \{0,1\}^n$ is sent through a probabilistic \emph{deletion channel} which independently deletes each bit with probability $\delta$ and concatenates the surviving bits, yielding a \emph{trace} of $x$. The problem is to reconstruct $x$ given independent traces. This problem has received much attention in recent years both in the worst-case setting  where $x$ may be an arbitrary string in $\{0,1\}^n$ \cite{DOS17,NazarovPeres17,HHP18,HL18,Chase19} and in the average-case setting where $x$ is drawn uniformly at random from $\{0,1\}^n$ \cite{PeresZhai17,HPP18,HL18,Chase19}.

This paper studies trace reconstruction in the \emph{smoothed analysis} setting, in which  a ``worst-case'' string $x^{\worst}$ is chosen arbitrarily from $\{0,1\}^n$, and then a  perturbed version $\bx$ of $x^{\worst}$ is formed by independently replacing each coordinate by a uniform random bit with probability $\sigma$. The problem is to reconstruct $\bx$ given independent traces from it.

Our main result is an algorithm which, for any constant perturbation rate $0<\sigma < 1$ and any constant deletion rate $0 < \delta < 1$, uses $\poly(n)$ running time and traces and succeeds with high probability in reconstructing the string $\bx$. This stands in contrast with the worst-case version of the problem, for which $\text{exp}(O(n^{1/3}))$ is the best known time and sample complexity \cite{DOS17,NazarovPeres17}.

Our approach is based on reconstructing $\bx$ from the multiset of its short subwords and is quite different from previous algorithms for either the worst-case or average-case versions of the problem. The heart of our work is a new $\poly(n)$-time procedure for reconstructing the multiset of all $O(\log n)$-length subwords of any source string $x\in \{0,1\}^n$ given access to traces of $x$.

\end{abstract}

%\newpage

%\input{abstract_nocomment_safe}

\newpage

\setcounter{page}{1}

%!TEX root = main.tex

\section{Introduction} \label{sec:intro}

\emph{Trace reconstruction} is a simple-to-state algorithmic problem which has been intensively studied yet remains mysterious in many respects.  The problem %seems to 
  captures some of the core algorithmic~challenges that arise in dealing with the \emph{deletion channel}; this is a noise process which, when given~an input string, independently deletes each coordinate with some fixed probability $\delta$ and outputs the concatenation of surviving coordinates.  In the trace reconstruction problem an algorithm is given access to independent \emph{traces} of a fixed unknown string $x \in \zo^n$, where a ``trace'' of $x$, denoted $\bz \sim \Del_\delta(x)$, is the string $\bz$ that results from passing $x$ through a deletion channel. The task is to use these traces to reconstruct the unknown string $x.$ 

Variants of the trace reconstruction problem have a long history, going back at least to \cite{Kalashnik73}.  The problem was studied on and off throughout the 2000s \cite{Lev01a,Lev01b,BKKM04,KM05,VS08,HMPW08,MPV14}, and has seen a renewed surge of recent interest over the past few years \cite{DOS17,NazarovPeres17,PeresZhai17,HPP18,HHP18,Chase19,BCFSSfocs19,BCSSrandom19,KMMP19,Narayanan20,HPPZ20}\ignore{\rnote{Did I miss any relevant references here? Should we take out our two population recovery papers and the Narayanan paper or leave them in? \textbf{Sandip:} I am in favour of keeping the pop rec papers in, and maybe add `... versions of the problem (and some generalizations of it).' at the end of the sentence.} \rnote{Different note about references:  visiting \href{https://arxiv.org/abs/1801.04783}{https://arxiv.org/abs/1801.04783}, it looks like in late April 2020 a new version of this was uploaded with Alex Zhai added as an author. I added this as a separate citation since just replacing \cite{HPP18} with it would give an inaccurate chronology.}} with the development of new algorithms and lower bounds for both the worst-case and average-case versions of the problem as well as various generalizations. Below we describe these two versions of the problem and recall the current state of the art for each of them.

\subsection{Prior work:  Worst-case and average-case trace reconstruction}

The original version of the trace reconstruction problem is the \emph{worst-case} version, in which the unknown string $x$ is an arbitrary (i.e.~adversarially chosen) string from $\zo^n$. This version of the problem has proved to be quite challenging; the first non-trivial result is due to Batu et al. \cite{BKKM04}, who gave a $\poly(n)$-time algorithm that uses $\poly(n)$ traces and succeeds when the deletion rate $\delta$ is very small, at most $n^{-1/2 - \eps}$ for any $\eps > 0$.  In \cite{HMPW08} Holenstein et al. gave an algorithm that runs in $\smash{\text{exp}(\tilde{O}(n^{1/2}))}$ time using $\smash{\text{exp}(\tilde{O}(n^{1/2}))}$ traces and succeeds for any $\delta$ bounded away from 1 by a constant. Simultaneous and independent works of De et al.~\cite{DOS17} and Nazarov and Peres~\cite{NazarovPeres17} gave an algorithm that improves the running time and sample complexity of \cite{HMPW08} to $\smash{\text{exp}(O(n^{1/3}))}$. In this same constant-$\delta$ regime, successively stronger lower bounds on the required sample complexity were given by \cite{MPV14,HL18}, culminating in a $\smash{\tilde{\Omega}(n^{3/2})}$ lower bound due to Chase \cite{Chase19}.

Another natural variant of the trace reconstruction problem is the \emph{average-case} version; in this variant the unknown string $x$ is assumed to be drawn uniformly at random from $\zo^n$, and the goal is for the algorithm to succeed with high probability over the random choice of $x$.  This problem variant is motivated both by the apparent difficulty of the worst-case problem and by the fact that in various application domains it may be overly pessimistic to assume that the input string $x$ is adversarially generated.  Much more efficient algorithms are known for the average-case problem: several  early works \cite{BKKM04,KM05,VS08} gave efficient algorithms that succeed for trace reconstruction of almost all $x \in \zo^n$ for various $o_n(1)$ deletion rates $\delta$, and \cite{HMPW08} gave an %average-case trace reconstruction 
  algorithm that runs in $\poly(n)$ time using $\poly(n)$ traces when $\delta$ is at most some sufficiently small constant.  More recent results of  Peres and Zhai \cite{PeresZhai17} and Holden et al. \cite{HPP18,HPPZ20}, which build on worst-case trace reconstruction results of \cite{DOS17,NazarovPeres17}, substantially improve on this, with \cite{HPP18,HPPZ20} giving an algorithm which uses $\smash{\text{exp}(O(\log^{1/3} n ))}$ traces to reconstruct a random $x\in \{0,1\}^n$ in $n^{1+o_n(1)}$ time when  the deletion rate is any constant bounded away from~$1$.

Summarizing the results described above, the current $\smash{\text{exp}(O(n^{1/3}))}$  state-of-the-art for worst-case trace reconstruction is exponentially higher than the current $\smash{\text{exp}(O( \log^{1/3} n ))}$ state-of-the-art for average-case trace reconstruction.  Given this substantial gap, it is natural to investigate intermediate formulations of the problem between the worst-case and average-case models.

\subsection{This work:  Smoothed analysis of trace reconstruction}

The well-studied  \emph{smoothed analysis} model, introduced by Spielman and Teng \cite{ST01}, provides a natural framework for interpolating between worst-case and average-case complexity. In smoothed analysis the input to an algorithm is obtained by applying a random $\sigma$-perturbation to a worst-case input instance; here $\sigma$ is a ``perturbation rate,''  which it is natural to scale so that $\sigma=1$ corresponds to a truly random instance and $\sigma=0$ corresponds to a worst-case instance.  By choosing intermediate settings of $\sigma$ it is possible to interpolate between worst-case and average-case problem variants.

We now give a detailed statement of the smoothed trace reconstruction problem that we consider.   First, a ``worst-case'' string $x^{\worst}$ is chosen arbitrarily from $\zo^n$, and then a randomly perturbed version $\bx$ of the string $x^{\worst}$ is formed by independently replacing each coordinate of $x^{\worst}$ by a uniform random bit with probability $\sigma$.\ignore{\snote{Removed the `$> 0$' restriction on $\sigma$ for the definition of the model. The very next sentence considers what happens when $\sigma=0$.}}  %$\sigma > 0$.
The goal is to reconstruct $\bx$ given access to independent traces drawn from $\Del_\delta(\bx).$ Note that when $\sigma=0$ this reduces to the worst-case trace reconstruction problem, and when $\sigma=1$ this reduces to the average-case problem.

As our main result, we give an algorithm for the smoothed trace reconstruction problem. For any initial string $x^{\worst},$ our algorithm can recover a $1-1/\poly(n)$ fraction of perturbed strings $\bx$ obtained from $x^{\worst}$ (for any $\poly(n)$) %Our algorithm runs 
in polynomial time for any constant perturbation rate $0 < \sigma \leq 1$ and any constant deletion rate $0 < \delta < 1$. 
%Our algorithm runs in polynomial time for any constant perturbation rate $0 < \sigma < 1$ and any constant deletion probability $0 < \delta < 1$.
More precisely, the main theorem we prove is the following:

\begin{theorem} [Polyomial time smoothed trace reconstruction] \label{thm:main}
Let $0 < \delta, \eta, \tau <1$ and $0 < \sigma \leq 1$. 
Let $x^\worst$ be an arbitrary and unknown string in $\zo^n$ and let $\bx$ be formed from $x^\worst$ by independently replacing each bit of $x^\worst$ with a uniform random bit from $\zo$ with probability $\sigma$. 

There is an algorithm with the following guarantee: with probability at least $1-\eta$ \emph{(}over the random generation of $\bx$ from $x^\worst$\emph{)}, it is the case that
the algorithm, given access to independent traces drawn from $\Del_\delta(\bx)$, outputs the string $\bx$ with
probability at least $1-\tau$ \emph{(}over the random traces drawn from $\Del_\delta(\bx)$\emph{)}. Its running time, as well as the number of traces it uses, is%\snote{Is it fine to state the time and sample complexity together, or does it give the impression that they are exactly same?}
\[
    \left(\frac{n}{\eta}\right)^{O\left(\frac{1}{\sigma(1-\delta)} \log \frac{2}{1-\delta}\right)} \log \frac{1}{\tau}.
\]

\end{theorem}

\begin{comment}
\begin{theorem} [Polyomial time smoothed trace reconstruction] \label{thm:main}
Let $0 < \sigma,\delta <1$.
Let $x^\worst$ be an arbitrary and unknown string in $\zo^n$ and let $\bx$ be formed from $x^\worst$ by independently replacing each bit of $x^\worst$ with a uniform random bit from $\zo$ with probability $\sigma$. 

There is an algorithm which, given access to independent traces drawn from $\Del_\delta(\bx)$, with~probability at least $1 - \tau$ \emph{(}over the random generation  of $\bx$ from $x^\worst$ and random traces drawn from $\Del_\delta(\bx)$\emph{)} outputs the string $\bx$. The running time of the algorithm is \rnote{The following is what we get just by blindly combining \Cref{lem:from-subwords-to-x} and \Cref{thm:subword-deck-reconstruction}, but this doesn't seem like the best thing to do because of the $\tau$ appearing in the exponent.  Shouldn't we instead consider the basic algorithm as running with $\tau = 0.01$, say, and then think about running it $\log(1/\tau)$ times and taking majority vote. This would get rid of the $\tau$ in the exponents, and would actually get rid of the $\tau$'s almost everywhere in the whole paper except in the statement of \Cref{thm:main}, which seems like a good thing.  

If you agree, let's do this?}
 \blue{
\[
\poly(n) \cdot (1-\delta)^{{\frac {-O(\log(n/\tau))}{\sigma(1-\delta)}}} \cdot \log (1/\tau).
\]}
and the number of traces it uses is
\[
\blue{(1-\delta)^{{\frac {-O(\log(n/\tau))}{\sigma(1-\delta)}}} \log (4kn/\tau)}. 
\]
\end{theorem}
\end{comment}

It is interesting that while the best currently known algorithms for the worst-case problem, corresponding to $\sigma=0$, require $\smash{\text{exp}(O(n^{1/3}))}$ time, for any constant perturbation rate we can solve the problem in a dramatically more efficient way. Intuitively, this shows that worst-case instances for trace reconstruction are ``few and far between,'' in the sense that even a small perturbation of such an instance typically makes it much easier to solve.
\ignore{
%(at least with currently known algorithms)\xnote{Maybe we can remove ``(at least with currently known algorithms)''? I am a bit unclear about what this refers to.}. \rnote{I'm fine with removing it. I think what it referred to is that maybe it is the case that both worst-case and smoothed versions of the problem have, say, an $O(n^2)$-trace, $O(n^3)$-time algorithm and hence are equally hard. But I agree it's a little unnecessary/confusing to say ``(at least with currently known algorithms)''.}
}

\subsection{Techniques}

Before describing our approach we briefly recall some of the methods used in prior work for the worst-case and average-case problems and discuss why these approaches do not seem applicable to the smoothed problem that we consider.

\medskip
\noindent
{\bf Worst-case algorithms.}
All of the known worst-case algorithms \cite{HMPW08,DOS17,NazarovPeres17} for deletion rates bounded away from 1 are ``mean-based,'' meaning that they only use (estimates of) the $n$ expected values $\smash{\E_{\by \sim \Del_\delta(x)}[\by_i]}$, $i=1,\dots,n$. The two papers \cite{DOS17,NazarovPeres17} both show~that mean-based algorithms can only succeed if they are given estimates of these expectations that are additively $\pm \hspace{0.04cm}\text{exp}(-\Omega(n^{1/3}))$-accurate, and hence mean-based algorithms must inherently use $\smash{\text{exp}(\Omega(n^{1/3}))}$ traces for the worst-case problem.  Inspection of  \cite{DOS17,NazarovPeres17} shows that these worst-case lower bounds  
%reconstructing arbitrary $x \in \zo^n$ 
  for mean-based algorithms 
in fact hold for a $1 - o_n(1)$ fraction of strings in $\zo^n$. Thus, the mean-based algorithmic approach of \cite{HMPW08,DOS17,NazarovPeres17} will not work~for our smoothed variant of the problem (and indeed our algorithm is not a mean-based algorithm).

\medskip
\noindent
{\bf Average-case algorithms.}
The average-case algorithms of \cite{PeresZhai17,HPP18,HPPZ20} work by aligning individual traces (and are not mean-based).  The analysis builds off~of some of the structural results established in \cite{DOS17,NazarovPeres17}, but also employs sophisticated probabilistic arguments which heavily depend on the randomness of the source string $x$ being reconstructed.  

As noted in \cite{HPP18,HPPZ20}, their average-case algorithm extends to the setting in which the target string $x$ is drawn from the $p$-biased distribution over $\zo^n$ (under which each bit $x_i$ is independently taken to be 1 with probability $p$).  Taking $p=\sigma/2$, this corresponds to our smoothed analysis model in the special case in which the original  string $x^\worst$ is promised to be the string $0^n$.  Equivalently, we can view our smoothed analysis problem as a more challenging variant of $p$-biased average-case trace reconstruction --- more challenging because the initial string ($x^{\worst}$) is no longer promised to be $0^n$, but rather \ignore{it\snote{remove the `it'?}}is both arbitrary and moreover unknown to the reconstruction algorithm. It is not  clear how to extend the $p$-biased average-case results of \cite{HPP18,HPPZ20} even to the setting in which the starting string $x^{\worst}$  is a \emph{known} arbitrary string, let alone to our setting in which $x^\worst$ is both arbitrary and unknown.

\medskip

\subsection{Our approach: Reconstruction from subwords and the subword deck\\ reconstruction problem}

In contrast with prior algorithms for the worst-case and average-case problem, our approach is based on first reconstructing \emph{subwords} of the target string and then reconstructing the target string from those subwords.  Recall that a subword of a string $x=(x_1,\dots,x_n)$ is a sequence of contiguous characters of $x$, i.e. a $(b-a+1)$-character string $(x_a,x_{a+1},\dots,x_b)$  for some $1 \leq a \leq b \leq n$.%\snote{Change the indexing to start from $0$?}.  

\medskip

\noindent {\bf Reconstruction from subwords.}
Given a length $1 \leq k \leq n$, let us write $\subword(x,k)$ to denote the \emph{multiset} of all $n-k+1$ length-$k$ subwords of $x$; we refer to this multiset as the \emph{$k$-subword deck of $x$}. For example, if $n=7$ and $k=3$, then the $k$-subword deck of $x=1101011$ would be the 5-element multiset $\{010, 011, 101, 101, 110\}$. 

In general the $k$-subword deck of $x$ may not uniquely identify the string $x$ within $\zo^n$ unless $k$ is very large; for example, the two multisets 
\[
\subword\big(0^{n/4}1^{n/4 - 1} 0^{n/4} 1^{n/4+ 1},k\big) \quad \text{and} \quad
\subword\big(0^{n/4}1^{n/4+1}0^{n/4}1^{n/4 - 1},k\big)
\]
are identical for every $k \leq n/4-1.$ This simple example shows that for worst-case strings $x$,  the $k$-subword deck of $x$ may not suffice to information-theoretically specify $x$ unless $k$ is linear in $n$.

The starting point of our approach is the observation that the situation is markedly better for \emph{random perturbations} of worst-case strings:  for any worst-case string $x^\worst \in \zo^n$, with high probability a random $\sigma$-perturbation $\bx$ of $x^\worst$ is such that $\subword(\bx,k)$ \emph{does} uniquely identify $\bx$ within $\zo^n$ even if $k$ is relatively small. Moreover, there is an efficient algorithm to reconstruct $\bx$ from $\subword(\bx,k)$. This is captured by the following result, which we prove in \Cref{sec:from-subwords-to-x}:

%\rnote{Following up on the earlier rnote about $\tau$: If we go the route proposed there, then we don't do anything interesting with the $\tau$ parameter at all anywhere. So how about we get rid of it in \Cref{lem:from-subwords-to-x} and in \Cref{thm:subword-deck-reconstruction} --- just hardwire it to $0.1$ in each of those. And add a statement after \Cref{thm:subword-deck-reconstruction} as indicated in the rnote there.}

\begin{lemma} [Reconstructing perturbed strings from their subword decks] \label{lem:from-subwords-to-x}
Let $0 < \sigma,\eta%\tau
< 1$.  There is a deterministic algorithm {\tt Reconstruct-from-subword-deck} which takes as input the $k$-subword deck $\subword(x,k)$ of a string $x \in \zo^n$, where $k = O({ {\log(n /\eta)}/{\sigma}})$, and outputs either a string in $\zo^n$ or \emph{``\textsf{fail}.''} {\tt Reconstruct-from-subword-deck} runs in $\poly(n)$ time and has the following property:  for any $x^\worst \in \zo^n$, if $\bx$ is a random $\sigma$-perturbation of $x^\worst$ (i.e. $\bx$ is obtained by independently replacing each bit of $x^\worst$ with a uniform random bit with probability $\sigma$), then with probability at least $1-\eta$ %$1-\tau/2$
the output of {\tt Reconstruct-from-subword-deck} on input $\subword(\bx,k)$ is the string $\bx$.
%Let $0 < \sigma,\tau < 1$.  There is an algorithm {\tt Reconstruct-from-subword-deck} which takes as input the $k$-subword deck $\subword(x,k)$ of a string $x \in \zo^n$, where $k = O({\frac {\log(n) + \log(1/\tau)}{\sigma}})$, and outputs a string in $\zo^n$.  {\tt Reconstruct-from-subword-deck} runs in $\poly(n)$ time and has the following property:  For any $x^\worst \in \zo^n$, if $\bx$ is a random $\sigma$-perturbation of $x^\worst$ (i.e. $\bx$ is obtained by independently replacing each bit of $x^\worst$ with a uniform random bit with probability $\sigma$), then with probability at least $1-\tau/2$ the output of {\tt Reconstruct-from-subword-deck} on input $\subword(\bx,k)$ is the string $\bx$.
\end{lemma}

\noindent 
{\bf The subword deck reconstruction problem.}  \Cref{lem:from-subwords-to-x} naturally motivates the algorithmic problem of \emph{subword deck reconstruction}:  given access to independent traces drawn from $\Del_\delta(x)$ and a length $k$, can we reconstruct the $k$-subword deck of $x$?  Our main algorithmic contribution is an efficient algorithm for this problem:

\begin{theorem} [Reconstructing the $k$-subword deck of $x$] \label{thm:subword-deck-reconstruction}
Let $0 < \delta,\tau < 1$. There is an algorithm {\tt Reconstruct-subword-deck} which takes as input a parameter $1 \leq k \leq n$ and access to independent traces of an unknown source string $x \in \zo^n$.
The running time of {\tt Reconstruct-subword-deck}, as well as the number of traces it uses, is
\[
\left(n \left( {\frac 2 {1-\delta}}\right)^k \right)^{O(1/(1-\delta))} \log \frac{1}{\tau}\ .
\]
%XXX(n,k,\delta,\tau)
{\tt Reconstruct-subword-deck} has the following property: for any string $x \in \zo^n$, with probability at least $1-\tau$ %$1-\tau/2$
the output of {\tt Reconstruct-subword-deck} is the $k$-subword deck $\subword(x,k)$.
\end{theorem}

\Cref{thm:main} follows immediately from \Cref{lem:from-subwords-to-x} and \Cref{thm:subword-deck-reconstruction}.
%\rnote{could add her ``by taking a simple majority vote over the $O(\log(1/\tau))$ candidate strings obtained by performing this many paired executions of {\tt Reconstruct-subword-deck}  and {\tt \tt Reconstruct-from-subword-deck}''}
We note that \Cref{thm:subword-deck-reconstruction} dominates the overall running time of \Cref{thm:main}, and that \Cref{thm:subword-deck-reconstruction} works for arbitrary strings.  %\in \zo^n$.

The algorithm in \Cref{lem:from-subwords-to-x}  and its analysis are relatively straightforward. To explain the main idea, we define the notion of the right (and left) extension of a string. (Starting from this point, 
  it will be convenient for us to index a string~$x \in \zo^n$
  using $\{0, \ldots,n-1\}$ as $x=(x_0,\dots,x_{n-1})$.) 
  
\begin{definition} \label{def:right-extension}
\hspace{-0.05cm}Given a $k$-bit string $w = (w_0,\dots,w_{k-1}) \in \zo^k$, a $k$-bit string $(w_1,\dots,w_{k-1},b)$~for some $b \in \{0,1\}$ is said to be a \emph{right-extension} of $w$.
We define \emph{left-extensions} of a string similarly.
\end{definition}

At a high level, the algorithm relies on the fact that if $\bx$ is obtained by a random $\sigma$-perturbation of $x^\worst$, then $\bx$ has useful \emph{local uniqueness properties}. More precisely, for $k = O({ {\log(n /\tau)}/{\sigma}})$, a simple probabilistic argument shows that with high probability $\bx[n-k:n-1]$ is the unique~element of $\subword(\bx,k)$  with no right-extension in $\subword(\bx,k)$. Consequently, we can identify $\bx[n-k:n-1]$ from the $k$-subword deck $\subword(\bx,k)$ of $\bx$. This argument can be extended inductively without much difficulty to in fact identify the whole of $\bx$.

%see Definition~\ref{def:right-extension} for t

%The analysis of \Cref{lem:from-subwords-to-x}  is based on relatively straightforward probabilistic arguments, and the algorithmic component is a simple induction based algorithm. The starting point is that for $k= 

 In contrast, \Cref{thm:subword-deck-reconstruction} is substantially more challenging. The structural results that underlie \Cref{thm:subword-deck-reconstruction} are based on two different sets of analytic arguments. The first argument only works when $\delta \le 1/2$ and employs (real) Taylor series; the second argument works for the entire range of $\delta <1$ and employs tools from complex analysis. While the first argument is more limited in scope of applicability, it is somewhat more elementary (which we see as a positive feature) and introduces a new ingredient (the so-called ``generalized deletion polynomial;'' see \Cref{sec:generalized}) which might be useful in future work, and thus  we include both arguments in the paper. In this proof overview below we only describe the second argument.

We begin by observing that $\subword(\bx,k)$ can be obtained by computing the multiplicity of occurrences of each $w \in \zo^k$ in the  set $\subword(\bx,k)$; we denote this multiplicity by $\#(w,\bx)$.
 The first key step  is to define a univariate polynomial (in the variable $\zeta$) $\SW_{\bx,w}(\zeta)$ which has the following two key properties: (i) $\SW_{\bx,w}(0)= \#(w,\bx)$, and (ii) using traces from $\Del_{\delta}(\bx)$, we have an unbiased estimator for $\SW_{\bx,w}(\zeta)$ for $\zeta= \delta$.
 Next, observe that given traces from $\Del_{\delta}(\bx)$, we can trivially simulate traces from $\Del_{\delta'}(\bx)$ for any $\delta' \ge \delta$, and hence we can get an unbiased estimator for $\SW_{\bx,w}(\zeta)$ for any $\zeta \in [\delta, 1]$. Recall, though, that our goal is to estimate $\SW_{\bx,w}(\zeta)$ at $\zeta =0$ and thus items (i) and (ii) above do not give us an unbiased estimator for $\SW_{\bx,w}(0)$.

 The most obvious idea at this point would be to do polynomial interpolation and use estimates for  $\SW_{\bx,w}(\zeta)$ for $\zeta \in [\delta, 1]$ to infer $\SW_{\bx,w}(0)$. Unfortunately, directly applying Lagrange interpolation is too naive an approach:  to accurately estimate $\SW_{\bx,w}(0)$, it turns out that we need $\SW_{\bx,w}(\zeta)$ for $\zeta \in [\delta, 1]$ up to error $\pm 2^{-\Theta(n)}$. However, to estimate $\SW_{\bx,w}(\zeta)$ up to error $\pm \kappa$, at least $\poly(1/\kappa)$ traces from $\Del_{\delta}(\bx)$
 are needed (\Cref{lem:estimator-polynomial}). Thus, directly applying Lagrange interpolation would require a sample complexity that grows like $2^{\Theta(n)}$, which is too expensive. 
 
Our approach is to forego Lagrange interpolation and instead (in essence) interpolate using tools from complex analysis. In particular, we prove a new structural result (\Cref{thm:complex}) about polynomials whose constant coefficient is not too small 
 and whose coefficients have magnitude bounded from above by a parameter $m$ (which is set to be $n^k$  in
 our application given that every coefficient of $\SW_{\bx,w}(\zeta)$ is bounded by $n^k$).
 %(as is the case for the polynomial $\SW_{\bx,w}(\zeta)$).  
 This result implies that $\SW_{\bx,w}(0)$ (which must be an integer given that $\SW_{\bx,w}(0)= \#(w,\bx)$)
    is uniquely determined by the values of $\SW_{\bx,w}(\zeta)$ in the interval $\zeta \in [\delta, 1]$ if these values are given up to error  $n^{O(-k/(1-\delta))}$; see \Cref{thm:uniqueness-info-theory}. Thus, in principle  we can determine  $\SW_{\bx,w}(0)$ by estimating $\SW_{\bx,w}(\zeta)$ for values of $\zeta \in [\delta, 1]$ to error $\pm n^{O(-k/(1-\delta))}$. This essentially implies that $\SW_{\bx,w}(0)$ can be determined using 
 $\approx n^{O( k/(1-\delta))}$ traces from $\Del_{\delta}(\bx)$.  (Note though that this sample complexity is not quite as good as is claimed in \Cref{thm:subword-deck-reconstruction}. We refine the above argument, using stronger coefficient bounds on $\SW_{\bx,w}$ and other ideas described at the beginning of \Cref{sec:strong-version}, to get  \Cref{thm:subword-deck-reconstruction} in its full strength as stated earlier.)
  
 In closing this subsection, we emphasize that while \Cref{thm:uniqueness-info-theory} is about the behavior of polynomials on the real line, its proof crucially uses tools from complex analysis such as Jensen's formula and the Hadamard three circle theorem. We further note that while we have sketched above how $\SW_{\bx,w}(0)$ can be determined in principle, this does not necessarily give an efficient algorithm.  To get an efficient algorithm, we use an approach based on linear programming.

 % in particular, for 
% $\zeta  \in [\delta, 1]$, to get $\SW_{\bx,w}(\zeta)$ up to error $\pm \kappa$, the unbiased estimator 

  %one employing Taylor series and one employing a new analytic lemma about extremal properties of polynomials with restricted coefficients which is proved using complex analysis. \blue{The algorithm based on these structural results is also more complex (an inductive approach with a linear programming subroutine).}\snote{algorithms? Or do we pretend for now that it's one algorithm?}

%\rnote{Augment/spruce up this next paragraph of prose describing some of the technical guts of what we do} At the heart of \Cref{thm:subword-deck-reconstruction} is a new structural result about polynomials saying \red{XXX}. In the case when the deletion rate $\delta$ is at most $1/2$, we are able to establish this using a general method based on \red{YYY}; this argument has the advantage of being general, and possibly extendible to other problems, but the drawback is that it only gives a result for $\delta \leq 1/2$. To handle deletion rates $\delta > 1/2$, we give a new structural result about extremal properties of polynomials with restricted coefficients (\red{explain in more detail}) which is proved using basic tools from complex analysis (the Hadamard Three Circles Theorem).

\subsection{Discussion and future work}

We view this paper as a first exploration, establishing that the algorithmic framework of smoothed analysis  can be fruitfully brought to bear on the trace reconstruction problem.  There are many interesting questions and directions for future work, some of which we highlight below.

One natural question is to establish strong sample complexity lower bounds for smoothed trace reconstruction.  Currently the best lower bound we are aware of for this framework is $\smash{\tilde{\Omega}( \log ^{5/2}n )}$ for average-case trace reconstruction due to \cite{Chase19}. Can an $\smash{n^{\Omega(1)}}$ lower bound be established for the smoothed model?

Another natural goal is to quantitatively strengthen our algorithmic result.  In the regime of $\sigma=c/n$ with $c$ a small constant, the smoothed problem reduces to the worst-case problem, and in the regime $\sigma = 1$ it reduces to the average-case problem; however, the running times of our algorithm in these regimes do not match the state-of-the-art running times for the corresponding problems that are provided in \cite{DOS17,NazarovPeres17} and in \cite{HPP18,HPPZ20} respectively. As a concrete first question along these lines, is it possible to improve the sample complexity of our algorithm from its current $n^{1/\sigma}$ dependence on the perturbation rate to a dependence more like $\smash{n^{1/\sigma^{1/3}}}$?

\ignore{
%\gray{Finally, the deletion noise channel that we have considered throughout this work can be augmented to also allow for probabilistic bit-flips and insertions of random bits. A number of algorithmic results \cite{DOS17,NazarovPeres17,HPPZ20} for trace reconstruction go through in this augmented model; we suspect that the same is true for our smoothed results but leave this as a subject for future work.}
}

%!TEX root = main.tex

\section{Preliminaries} \label{sec:preliminaries}

\noindent {\bf Notation.}
Given a nonnegative integer $n$, we write $[n]$ to denote $\{1,\ldots,n\}$.
Given integers $a\le b$ we write $[a:b]$ to denote $\{a,\ldots,b\}$. 
It will be convenient for us to index a binary string~$x \in \zo^n$
  using $[0:n-1]$ as $x=(x_0,\dots,x_{n-1})$. Given such a string $x$ and integers $0 \leq a < b \leq n-1$, we write $x[a:b]$ to denote the \emph{subword} $(x_a, x_{a+1}, \dots, x_b)$.
%Given a vector $v=(v_1,\ldots,v_d) \in \R^d$, we write $\|v\|_\infty$ to denote $\max_{i \in [d]} |v_i|.$
%Given a function $\Delta:A\rightarrow \R$ over a finite domain $A$, we write $\|\Delta\|_\infty=\max_{a\in A} |\Delta(a)|$.
%Given a polynomial $p$ (which may be univariate or multivariate), we write $\|p\|_1$ to denote the sum of magnitudes of $p$'s coefficients. 
We write $\ln$ to denote natural logarithm and $\log$ to denote logarithm to the base 2.

%We use $\ast^m$ as a placeholder to stand for  any string in $\zo^m$.
We denote the set of non-negative integers by $\Z_{\geq 0}$.
For a vector $\alpha = (\alpha_1, \dots, \alpha_{\ell}) \in \Z_{\geq 0}^{\ell}$, we write $|\alpha|$ to denote $\alpha_1 + \alpha_2 + \dots + \alpha_{\ell}$, and write $\alpha!$ to denote $\alpha_1! \alpha_2! \cdots \alpha_{\ell}!$.

\medskip

\noindent {\bf Subword deck.} Fix a string $x \in \zo^n$ and an integer $k \in [n]$. A $k$-subword of $x$ is a %contiguous substring of 
 (contiguous) subword of $x$ of length $k$, given by $(x_a, x_{a+1}, \dots, x_{a+k-1})$ for some $a \in [0: n-k]$. For a string $w \in \zo^k$, let $\#(w,x)$ denote the number of occurrences of $w$ as a subword of $x$. We define the \emph{$k$-subword deck} of $x$, denoted $\subword(x,k)$, to be the $(n-k+1)$-size (unordered) multiset of all $k$-subwords~of $x$.
 We also extend the notation of $\#(w,x)$ to strings $w\in \{0,1,*\}^k$, where 
   $*$ is the wildcard symbol:
   $\#(w,x)$ is the sum of $\#(w',x)$ over all $w'\in \{0,1\}^k$ with $w_i'=w_i$ for every $w_i\ne *$. 

\medskip

\noindent {\bf Distributions.}
We use bold font letters to denote probability distributions and 
  random variables, which should be clear from the context.
We write ``$\bx \sim \bX$'' to indicate that random variable~$\bx$~is 
  distributed according to distribution $\bX$.

\medskip

\noindent {\bf Deletion channel and traces.}
Throughout this paper the parameter $\delta:0 <$ $\delta < 1$ denotes~the \emph{deletion probability}.  Given a string $x \in \zo^n$, we write $\Del_\delta(x)$ to denote the distribution of the string that results from passing  $x$ through the $\delta$-deletion channel (so the distribution $\Del_\delta(x)$ is supported on $\zo^{\leq n}$), and we refer to a string in the support of $\Del_\delta(x)$ as a \emph{trace} of $x$.  Recall that a random trace $\by \sim \Del_\delta(x)$ is obtained by independently deleting each bit of $x$ with probability $\delta$ and concatenating the surviving bits.\hspace{0.05cm}\footnote{For simplicity in this work we assume that the deletion probability $\delta$ is known to the reconstruction algorithm.  We~note that it is possible to obtain a high-accuracy estimate of $\delta$ simply by measuring the average length of traces received from the deletion channel.}

\medskip

\noindent {\bf Perturbation and smoothed analysis.}. The perturbation model we consider corresponds to the standard notion of perturbation of an $n$-bit string which arises in the analysis of Boolean functions.  Given an $n$-bit string $x^\worst \in \zo^n$, a \emph{$\sigma$-perturbation of $x^\worst$} is a random string $\bx \in \zo^n$ obtained by independently setting each coordinate $\bx_i$ to be $x^\worst_i$ with probability $1-\sigma$ and to be uniformly random with the remaining probability $\sigma$.  Equivalently, $\bx$ is a random string that is $(1-\sigma)$-correlated with $x^\worst$; in the notation of Chapter~2 of \cite{Odonnell-book}, we may write this as ``$\bx \sim N_{1-\sigma}(x^\worst)$.''

We recall that in the smoothed analysis framework, an initial string
  $x^\worst \in \zo^n$ is selected (in what may be thought of as an adversarial manner), and then a $\sigma$-perturbation $\bx$ of $x^\worst$ is drawn at random from $N_{1-\sigma}(x^\worst)$, and the algorithm runs on instance $\bx$.  The goal is to develop algorithms which, for every $x^\worst \in \zo^n$, succeed with high probability on the perturbed instance $\bx \sim N_{1-\sigma}(x^\worst).$

%\blue{We introduce an additional parameter: the perturbation rate\snote{This definition needs editing to be consistent with the introduction.} %smoothing rate\snote{Better name? \red{Rocco:} maybe ``perturbation rate''?} 
%$\sigma \in (0,1/2]$
%\snote{This is w.l.o.g. since, if $\sigma > 1/2$, an equivalent view is that $\1-x^{\worst}$ is smoothed at rate $1-\sigma \leq 1/2$.}

%Given an arbitrary string $x^{\worst}$, a random %(unknown) string $\bx \in \zo^n$ is generated from $x^{\worst}$ by flipping each of its bits independently with probability $\sigma$; this is denoted by $\bx \sim \Flip_{\sigma}(x^{\worst})$.}

%!TEX root = main.tex

\section{Reconstructing perturbed worst-case strings from their\\ subword decks: 
Proof of \Cref{lem:from-subwords-to-x}} \label{sec:from-subwords-to-x}

In this section we prove \Cref{lem:from-subwords-to-x}:

\medskip

 \noindent {\bf Restatement of \Cref{lem:from-subwords-to-x}} (Reconstructing perturbed strings from their subword decks){\bf .}
\emph{
Let $0 < \sigma,\eta < 1$.  There is a deterministic algorithm {\tt Reconstruct-from-subword-deck} which takes as input the $k$-subword deck $\subword(x,k)$ of a string $x \in \zo^n$, where $k = O({ {\log(n /\eta)}/{\sigma}})$, and outputs either a string in $\zo^n$ or \emph{``\textsf{fail}.''} {\tt Reconstruct-from-subword-deck} runs in $\poly(n)$ time and has the following property:  For any string $x^\worst \in \zo^n$, if $\bx$ is a random $\sigma$-perturbation of $x^\worst$ (i.e. $\bx$ is obtained by independently replacing each bit of $x^\worst$ with a uniform random bit with probability $\sigma$), then with probability at least $1-\eta$ %/2$ 
the output of {\tt Reconstruct-from-subword-deck} on input $\subword(\bx,k)$ is the string $\bx$.
}

\medskip

The idea of \Cref{lem:from-subwords-to-x} is very simple:  a probabilistic argument shows that for any worst-case string $x^\worst$, a random $\sigma$-perturbation introduces enough variability into $\bx \sim N_{1-\sigma}(x^\worst)$ so that the $k$-subwords comprising the $k$-subword deck of $\bx$ can be easily pieced together in a unique way to yield $\bx$ by a simple greedy algorithm. We now provide details.

%Let $\bx \sim \Flip_\sigma(x^{\worst})$, where $x^{\worst}$ is an arbitrary $n$-bit string. We first show that the $k$-subword deck of $\bx$, along with the initial $k$-subword $\bx[0:k-1]$, uniquely determines $\bx$ with high probability.

Given $\subword(x,k)$ of a string $x\in \zo^n$,
  we use the following greedy algorithm to recover $x$:
\begin{flushleft}\begin{enumerate}
\item %Let $S=\subword(x,k)$, a multiset of $n-k+1$ strings of length $k$ each.
We will store the output in $y$, a string of length $n$. 
\item Let $w\in \subword(x,k)$ be a string that fails to have a right-extension in 
  $\subword(x,k)$.
(Note the only $k$-subword of $x$ that can fail to have a right-extension in $\subword(x,k)$
  is $x[n-k:n-1]$.)
If no such $w$ exists, return \textsf{fail}; otherwise set $y[n-k:n-1]=w$ and $\ell=n-k$.
\item While $\ell>0$, do the following: 
  Find $w\in \subword(x,k)$ as a left-extension of $y[\ell:\ell+k-1]$.
  (Note that if $y$ agrees with $x$ so far, then such a left-extension must exist.)
  If $w$ is not unique (counted with multiplicity), return \textsf{fail}; otherwise set $y_{\ell-1}=w_0$ and decrement $\ell$ by 1.
\item When $\ell=0$, return $y$.
\end{enumerate}\end{flushleft}

It is clear from the description of the greedy algorithm above and comments therein that
  either it returns \textsf{fail} or there is no ambiguity (in filling in the last $k$ bits
  and extending from there bit by bit) and $x$ is recovered correctly as $y$ at the end.
We use the following definition to capture strings $x$ on which the greedy algorithm succeeds:
%The only way that a string $w \in \subword(x,k)$ can fail to have a right-extension in $\subword(x,k)$ is if $w=x[n-k:n-1]$ and there is no other occurrence of $w$ as an earlier subword in $x$. On the other hand, it is possible for a string $w \in SW_k(x)$ to have both possible right-extensions $(w_1,\dots,w_{k-1},0)$ and $(w_1,\dots,w_{k-1},1)$ be present in $SW_k(x)$; these are the situations in which there is ambiguity as to how a greedy algorithm should reconstruct $x$ from $\subword(x,k)$. Strings $x$ in which such an ambiguity does not arise are said to be \emph{$k$-good}:

\begin{definition} \label{def:good}
     An $n$-bit string $x$ is said to be \emph{$k$-good} if 
\begin{flushleft}\begin{itemize}
\item[(i)] for every $j\in [n-k]$, there is exactly one string in $\subword(x,k)$ (counted
  with multiplicity) that is a left-extension of the subword $x[j:j+k-1]$; and
%for every $j \in [0:n-k-1]$, there is exactly one element among the $n-k+1$ many elements (counted with multiplicity) in $\subword(x,k)$ that is a right-extension of the subword $(x_j, x_{j+1},\dots,x_{j+k-1})$; and 
\item [(ii)] the subword $x[n-k:n-1]$ does not have a right-extension in $\subword(x,k).$
\end{itemize}\end{flushleft}
%and In other words, exactly one of $(x_{j+1},\dots,x_{j+k-1},0)$ and $(x_{j+1},\dots,x_{j+k-1},1)$ is in $SW_k(x)$.%In other words, either $(x_{j+1},\dots,x_{j+k-1},0)$ is in $SW_k(x)$ or $(x_{j+1},\dots,x_{j+k-1},1)$ is in $SW_k(x)$ but not both.
\end{definition}

%It is clear from \Cref{def:good} that if $x \in \zo^n$ is $k$-good, then a simple greedy algorithm can efficiently reconstruct $x$ from $\subword(x,k)$. (Given any element of $\subword(x,k)$, if it has no right-extension in $\subword(x,k)$ then it must be the $k$-bit suffix $(x_{n-k},\dots,x_{n-1})$ of $x$, and otherwise it can be aligned with its unique right-extension in $\subword(x,k)$ in the obvious way.) Thus 
To prove
\Cref{lem:from-subwords-to-x}, it remains only to establish the following claim:

%The motivation for this definition is that a $k$-good string $x$ is uniquely determined given its $k$-subword deck, along with the initial $k$-subword $\bx[0:k-1]$. Moreover, this reconstruction admits an efficient procedure: the greedy algorithm which starts with $\bx[0:k-1]$ and iteratively aligns the current $k$-subword with its unique right-extension is clearly correct, and it runs in polynomial time.\snote{There is a subtle point here: most strings being $k$-good means that if I get the $k$-subword deck of a $k$-good string $\bx^1$, the reconstruction would never give some other $k$-good string $\bx^2$. On the other hand it is not obvious from this that the reconstruction cannot return some $\bx'$ that is NOT $k$-good. That follows from the fact that the $k$-subword deck of $\bx'$ cannot match $\bx$'s by definition of $k$-goodness. Is this worth describing explicitly for the correctness guarantee?} Now we show that most strings obtained by the random perturbation are $k$-good for \blue{reasonably small} $k$. 

\begin{claim} \label{claim:subword-unique}
    Fix any string $x^{\worst} \in \zo^n$.  Then for $k = O ( \log (n/\eta) /{\sigma} )$%O({\frac {\log(n) +  \log(1/\eta)}{\sigma}})$, we have
    \[
\Prx_{\bx \sim N_{1-\sigma}(x^{\worst})}\big[\hspace{0.03cm}\bx\text{~is~$k$-good}\hspace{0.08cm}\big] \geq 1 - \eta.%/2.
\] 
%    the string $\bx$ is $k$-good for $k=7 \log n/\sigma$ with probability at least $1-1/n$ over $\bx$.%, where $C > 0$ is a sufficiently large constant.
\end{claim}

\begin{proof}
    Let $E(\bx)$ be the event that $\bx$ is \emph{not} $k$-good. We observe that for $E(\bx)$ to occur, there must exist indices $0 \leq i < j \leq n-k+1$ such that the $(k-1)$-subwords of $\bx$ starting at positions $i$ and $j$ are equal, %\snote{For clarity, I would have liked to keep the description that the $(k-1)$-subwords are equal but the $k$-subword extensions are not equal, even though we don't use the stronger condition.} %, but the $k$-length subwords starting at the same positions are not,
     i.e., $\bx[i:i+k-2] = \bx[j:j+k-2]$. % but $\bx_{i+k-1} \neq \bx_{j+k-1}$. 
     In particular, we have the following (where here and subsequently all probabilities are over the random draw of $\bx \sim N_{1-\sigma}(x^\worst)$):
    \begin{align*}
        \Pr\big[E(\bx)\big] &\leq \Pr\Big[\hspace{0.03cm}\exists\hspace{0.06cm}i,j \text{ such that } \bx[i:i+k-2] = \bx[j:j+k-2]\hspace{0.02cm}\Big]\\[0.6ex]
        &\leq \sum_{0\leq i < j \leq n-k+1} \Pr\Big[\hspace{0.05cm}\bx[i:i+k-2] = \bx[j:j+k-2]\hspace{0.02cm}\Big]. \tag{by a union bound}
    \end{align*}
    Let $E_{i,j}(\bx)$ denote the event that $\bx[i:i+k-2] = \bx[j:j+k-2]$. To prove the claim, it suffices to show that $\Pr[E_{i,j}(\bx)] \leq \eta/n^2$ for each fixed pair $1 \leq i < j \leq n-k+1$.
    
To this end, we write the probability of $E_{i,j}(\bx)$ as 
$$
\Pr\big[\bx_i=\bx_j\big]\cdot \prod_{\ell=1}^{k-2}\hspace{0.1cm}
\Pr\Big[\hspace{0.04cm}\bx_{i+\ell}=\bx_{j+\ell}\hspace{0.06cm}\Big|\hspace{0.06cm} \bx_{i+h}=\bx_{j+h}\ \text{for all $h=0,\ldots,\ell-1$}\hspace{0.04cm}\Big].
$$    
The first factor $\Pr\big[\bx_i=\bx_j\big]$ is at most $1-\sigma/2$ because for any
  fixed value $b$ of $\bx_i$, $\bx_j$ agrees with $b$ after the perturbation with probability 
  at most $1-\sigma/2$.
The upper bound of $1-\sigma/2$ holds for every other factor in the product. 
For the $\ell$th factor, we note that for any fixed values of 
  $\bx_{i},\ldots,\bx_{j+\ell-1}$ that satisfy the conditioning part 
  $\bx_{i+h}=\bx_{j+h}$ for all $h=0,...,\ell-1$,
  $\bx_{j+\ell}$ agrees with the fixed value of $\bx_{i+\ell}$
  with probability at most $1-\sigma/2$.
Thus, by setting $k=C\log (n/\eta) / \sigma$ for some large enough constant $C$, we have 
    \[
        \Pr\big[E_{i,j}(\bx)\big] \leq (1-\sigma/2)^{k-1} \leq \exp\left(-\Omega\left(\log \frac{n}{\eta}\right)\right) \leq \frac{\eta}{n^2}.
    \]
This finishes the proof of the claim.
%    for a suitable choice of the constants hidden in the big-Oh notation of the upper bound on $k$, and the claim is proved.
\end{proof}

\ignore{
%Henceforth, we assume that $\bx$ is $k$-good for $k=7\log n/\sigma$. Now we give $2$ procedures to reconstruct the $k$-subword deck of $\bx$ using $\exp(O(k)) \poly(n)$ traces, where the first one requires an additional assumption that the deletion probability $\delta \leq 1/2$, while the second requires no such assumption.
}

%!TEX root = main.tex

\section{Reconstructing the $k$-subword deck:  Proof of \Cref{thm:subword-deck-reconstruction}} \label{sec:proof-of-thm-subword-deck-reconstruction}
%Proof 1: General deletion-channel polynomials when $\delta \leq 1/2$}

The remaining task to establish the main result, \Cref{thm:main}, is to prove \Cref{thm:subword-deck-reconstruction} \blue{(restated below)}, which gives an efficient algorithm to reconstruct the $k$-subword deck of an arbitrary source string $x \in \zo^n$ given access to independent traces of $x$. \blue{Throughout this section, let $\rho = (1-\delta)/2$.}%\snote{Replace $1-\delta$ by $(1-\delta)/2$ everywhere.}

\medskip

\noindent {\bf Restatement of \Cref{thm:subword-deck-reconstruction}} (Reconstructing the $k$-subword deck){\bf .}
\emph{
Let $0 < \delta,\tau' < 1$.  There is an algorithm {\tt Reconstruct-subword-deck} which takes as input a parameter 
$1 \leq k \leq n$
and access to independent traces of an unknown source string $x \in \zo^n$.
The running time of {\tt Reconstruct-subword-deck}, as well as the number of traces it uses, is
\[
\left( n/{\rho^k}\right)^{O(1/\rho)}  \log (1/\tau').
\]
{\tt Reconstruct-subword-deck} has the following property:  For any unknown string $x \in \zo^n$, with probability at least $1-\tau'$, %$1-\tau'/2$, 
{\tt Reconstruct-subword-deck}
  outputs $\subword(x,k)$.
}

\medskip

The main algorithmic ingredient that underlies \Cref{thm:subword-deck-reconstruction} is an algorithm for a closely related but slightly simpler problem. This algorithm, which we call {\tt Multiplicity}, takes as input a string $w \in \zo^k$ and access to independent traces from an unknown source string $x$, and it outputs $\#(w,x)$, the multiplicity of $w$ in the $(n-k+1)$-element multiset $\subword(x,k)$ (note that this multiplicity can be zero if $w$ is not present as a subword of $x$):

\begin{theorem} \label{thm:Multiplicity}
Let $0 < \delta,\tau < 1$ and let $\rho = (1-\delta)/2$. There is an algorithm {\tt Multiplicity} which takes as input a string $w \in \zo^k$ %\blue{, where $\log n \leq k \leq n$,}\snote{This is to simplify the bound; if $k < \log n$ then the sample/time complexity will have terms like $k+\log n$ in the exponent, and also the analysis in Sections \ref{sec:improvement_technical}, \ref{sec:improvement_small_delta} will get more involved and involve cases. Since we only ever use it for $k\geq \log n$ anyway, I think it's better to just impose this here, but I'm happy to change this.}
and access to independent traces of an unknown source string $x \in \zo^n$.
{\tt Multiplicity} runs in \blue{$\left(n / \rho^k\right)^{O(1/\rho)} \log (1/\tau)$} %\blue{$\left(n / (1-\delta)^k\right)^{O(1/\rho)} \log (1/\tau)$} 
time and uses \blue{$\left(n / \rho^k\right)^{O(1/\rho)} \log (1/\tau)$}
%\blue{$\left(n / (1-\delta)^k\right)^{O(1/\rho)} \log (1/\tau)$} %$\red{BBB(n,k,\delta,\tau)}$
many traces from $\Del_\delta(x)$, and has the following property:  For any unknown source string $x \in \zo^n$, with probability at least $1-\tau$   %$1-\tau/2$
the output of {\tt Multiplicity} is $\#(w,x)$ \emph{(}i.e.~the number of occurrences of $w$ as a subword of $x$\emph{)}.
\end{theorem}

%\cnote{Do we have two algorithms for {\tt Multiplicity}?  One for $\delta \le 1/2$ and one for arbitrary $\delta$?}\rnote{Yes - at this point in the exposition though I think we can just state it as one algorithm and explain later.}
A standard ``branch-and-bound'' argument gives \Cref{thm:subword-deck-reconstruction} from \Cref{thm:Multiplicity}:

\medskip

\noindent \emph{Proof of \Cref{thm:subword-deck-reconstruction} using \Cref{thm:Multiplicity}.}
Let $\ell=\lfloor \log n\rfloor$. We first consider the case that $k \leq \ell $.  In this case {\tt Reconstruct-subword-deck} simply runs {\tt Multiplicity}$(w)$ once for each of the $2^k$ strings $w \in \zo^k$, with the confidence parameter ``$\tau$'' for each run of {\tt Multiplicity} set to $\tau'/2^k$. Since we can reuse the same traces for each of the $2^k$ runs, in this case the running time is \blue{$2^k \left(n / \rho^k\right)^{O(1/\rho)} \log (2^k/\tau') = \left(n / \rho^k\right)^{O(1/\rho)} \log (1/\tau')$}  %\red{$2^k \cdot AAA(n,k,\delta,\tau'/2^{k+1}) \leq O(kn) \cdot AAA(n,k,\delta,\tau'/(4kn))$}
and the sample complexity is
\blue{$\left(n / \rho^k\right)^{O(1/\rho)} \log (1/\tau')$}.
%\left(n / (1-\delta)^k\right)^{O(1/\rho)} \log (2^k/\tau') =
%\left(n / (1-\delta)^k\right)^{O(1/\rho)} \log (1/\tau')$.} 
%\red{$BBB(n,k,\delta,\tau'/2^{k+1}) \leq BBB(n,k,\delta,\tau'/(4kn))$}.

Next we consider the case that $k > \ell$.  To avoid an exponential running time dependence on $k$, the algorithm uses a simple ``branch-and-prune'' approach.  In the first stage, similar to the previous paragraph, {\tt Reconstruct-subword-deck} runs {\tt Multiplicity} on each of the $2^\ell$ strings $w \in \zo^{\ell}$ with confidence parameter $\tau'/(2nk)$, thereby obtaining the $\ell$-subword deck $\subword(x,\ell)$.  It then executes $k - \ell$ many successive stages $j=1,2,\dots,k-\ell,$ where in stage $j$ the algorithm determines the $(\ell + j)$-subword deck of $x$ using the $(\ell + j-1)$-subword deck of $x$. It does this in each stage as follows:  for each of the (at most $n$) distinct strings $w \in \subword(x,\ell + j - 1)$, the algorithm runs {\tt Multiplicity}$(w0)$ and {\tt Multiplicity}$(w1)$, each with confidence parameter $\tau'/(2nk)$.

The correctness of this approach follows from the trivial fact that an $(\ell+j)$-bit string can only be present in $\subword(x,\ell+j)$ if its $(\ell+j-1)$-bit prefix is present in $\subword(x,\ell+j-1)$. Since there are at most %\red{
$n + 2n(k-\ell) < 2kn$
%}
many runs of {\tt Multiplicity} overall, the running time of {\tt Reconstruct-subword-deck} is at most \blue{$O(kn) \cdot \left(n / \rho^k\right)^{O(1/\rho)} \log (2kn/\tau') = \left(n / \rho^k\right)^{O(1/\rho)} \log (1/\tau')$} 
%\left(n / (1-\delta)^k\right)^{O(1/\rho)} \log (2kn/\tau') = \left(n / (1-\delta)^k\right)^{O(1/\rho)} \log (1/\tau')$} 
%\red{$O(kn) \cdot AAA(n,k,\delta,\tau'/(4kn))$}
and the sample complexity is at most
\blue{$\left(n / \rho^k\right)^{O(1/\rho)} \log (1/\tau')$}, 
%\left(n / (1-\delta)^k\right)^{O(1/\rho)} \log (2kn/\tau') =
%\left(n / (1-\delta)^k\right)^{O(1/\rho)} \log (1/\tau')$},
%\red{$BBB(n,k,\delta,\tau'/(4kn))$}
and \Cref{thm:subword-deck-reconstruction} is proved. \qed

\medskip

Thus, in the rest of the paper, we focus on proving \Cref{thm:Multiplicity}. 

\subsection{The subword polynomial}

The following ``subword polynomial'' plays an important role in our approach:

\begin{definition} \label{def:subword-polynomial}
Given $x \in \zo^n$ and $w=(w_0,\dots,w_{k-1}) \in \zo^k$, let $\SW_{x,w}(\zeta)$ be the following univariate polynomial of degree $n-k$:
\[
    \SW_{x,w}(\zeta) := \sum_{\mathclap{\substack{\alpha \in \Z_{\geq 0}^{k-1} \\ |\alpha| \leq n-k}}} \, \#\left(w_0 \ast^{\alpha_1} w_1 \ast^{\alpha_2} w_2 \cdots w_{k-2} \ast^{\alpha_{k-1}} w_{k-1},\hspace{0.06cm} x\right) \, \cdot\zeta^{|\alpha|}.
\]
\end{definition}

In words, the degree-$\ell$ coefficient of the subword polynomial $\SW_{x,w}$ is the number of ways that $w$ arises as a substring of $x$ with a total of exactly $\ell$ extraneous additional characters interspersed among the characters of $w$. In particular, we have that the constant term of $\SW_{x,w}$ (i.e. $\SW_{x,w}(0)$, since $0^0=1$) is equal to $ \#(w,x)$, the frequency of $w$ as a subword of $x$, which is what \Cref{thm:Multiplicity} aims to estimate efficiently from traces of $x$. 

\subsection{Outline of our approach}

We prove \Cref{thm:Multiplicity} by giving two different algorithms depending on the value of the deletion rate $\delta$.
The first of these algorithms, {\tt Multiplicity}$_{\text{small-}\delta}$, gives a simple and direct approach to compute the value $\SW_{x,w}(0) =\#(w,x)$; however this approach requires the deletion rate $\delta$ to be less than $1/2$. 
This approach is based on analyzing a new object, the ``generalized deletion polynomial,'' that we believe may be useful for subsequent work.  
The second of these algorithms, {\tt Multiplicity}$_{\text{large-}\delta}$, gives a different and somewhat more involved algorithm (involving linear programming and a new extremal result on polynomials, proved using complex analysis) that can be used for any deletion rate $\delta<1$.  

Readers who are  interested in a simple analysis (albeit one that works only for $\delta <1/2$)
  may wish to focus on {\tt Multiplicity}$_{\text{small-}\delta}$ (\Cref{sec:small-delta}).
Readers who are  interested in a more involved approach that succeeds for all $\delta<1$ may wish to focus on {\tt Multiplicity}$_{\text{large-}\delta}$ (\Cref{sec:large-delta}).  The two algorithms and analyses are each self-contained; each may be read independently of the other.

For each of the two algorithms, %{\tt Multiplicity}$_{\text{small-}\delta}$ and {\tt Multiplicity}$_{\text{large-}\delta}$, 
  we first give a simpler version of the analysis which establishes a quantitatively weaker version of the result, with an $n^{O(k)}$ running time and sample complexity (ignoring the dependence on other parameters); see the statements of \Cref{thm:Multiplicity-small-delta-weak} and \Cref{thm:Multiplicity-large-delta-weak}, at the beginnings of \Cref{sec:small-delta} and \Cref{sec:large-delta} respectively, for detailed statements of these weaker versions.  In \Cref{sec:strong-version} we quantitatively strengthen both  \Cref{thm:Multiplicity-small-delta-weak} and \Cref{thm:Multiplicity-large-delta-weak} to achieve a $\poly(n) \cdot \exp(O(k))$ running time and sample complexity, and thereby complete the proof of \Cref{thm:Multiplicity}.

%!TEX root = main.tex

\section{{\tt Multiplicity}$'_{\text{small-}\delta}$:  An algorithm for deletion rate $\delta <1/2$
}
\label{sec:small-delta}

In this section we prove \Cref{thm:Multiplicity-small-delta-weak},
  a weaker version of \Cref{thm:Multiplicity}.
It gives an algorithm that has $n^{O(k)}$ running time and sample complexity (ignoring the dependence on other parameters) and works when $\delta < 1/2$.
Actually \Cref{thm:Multiplicity-small-delta-weak} works when $\delta \le 1/2$;
  we only require $\delta < 1/2$ later in \Cref{sec:improvement_small_delta} %Section~\ref{} 
  to achieve
  the improved running time and sample complexity in \Cref{thm:Multiplicity} based on a similar approach (the running time achieved in that section will depend on how close $\delta$ is to $1/2$).
  
  %(it will become clear at the end of Section  $\delta\le 1/2$ for \Cref{thm:Multiplicity-small-delta-weak}; for its stronger version 
%  of \Cref{thm:Multiplicity} in Section \ref{} we require $\delta$ to be at most $1/2-c$ for some
%  fixed $c>0$):
%$ is bounded below $1/2$ by some constant:
%\rnote{Will the algorithm of this section work all the way up to $\delta = 1/2$? Or do we need $\delta$ to be at most, say, 0.49, for it to be efficient? I assumed the latter in the prose in the previous section}

\begin{theorem} \label{thm:Multiplicity-small-delta-weak}
Let $0 < \delta \le 1/2$. 
%where $c>0$ is an absolute constant.  
There is an algorithm  {\tt Multiplicity}$'_{\text{small-}\delta}$ which takes as input a string $w \in \zo^k$, access to independent traces of an unknown source string $x \in \zo^n$, and a parameter $\tau>0$.
{\tt Multiplicity}$'_{\text{small-}\delta}$ draws $n^{O(k)}  \cdot \log (1/\tau)$ traces from $\Del_\delta(x)$, runs in time $n^{O(k)}  \cdot \log (1/\tau)$,  and has the following property:  For any unknown source string $x \in \zo^n$, with probability at least $1-\tau$ %$1-\tau/2$
the output of {\tt Multiplicity}$'_{\text{small-}\delta}$ is the multiplicity of $w$ in $\subword(x,k)$ \emph{(}i.e.~the number of occurrences of $w$ as a subword of $x$\emph{)}. 
\end{theorem}

In \Cref{sec:strong-version} we will build on \Cref{thm:Multiplicity-small-delta-weak} to give a stronger version that has $\poly(n) \cdot \exp(O(k))$ running time and sample complexity (ignoring the dependence on other parameters) for $\delta < 1/2$.

The rest of this section is organized as follows. In \Cref{sec_estimate_poly_values}, we give an equivalent expression for $\SW_{x,w}(\zeta)$ in \Cref{thm:Taylor-subword-poly},
  which relates the subword polynomial to traces drawn from the deletion channel.
The proof uses the \emph{generalized deletion polynomial} and is presented in
  \Cref{sec:generalized}.  
This new expression for $\SW_{x,w}(\zeta)$ allows one to 
  evaluate $\SW_{x,w}(\zeta)$ at $\zeta=0$   
  up to a small error (say, $\pm 0.1$) using traces of $x$ (see Corollary \ref{corr-poly-eval})
  when $\delta\le 1/2$.
Given that $\SW_{x,w}(0)$ is an integer, 
  the result can be rounded to obtain the exact value of $\SW_{x,w}(0)$;
  this finishes the proof of \Cref{thm:Multiplicity-small-delta-weak}. 
  
We remark that the expression for $\SW_{x,w}(\zeta)$ given in \Cref{thm:Taylor-subword-poly}
  works for any $\zeta\in \C$, when viewing  $\SW_{x,w}(\zeta)$ as a polynomial over $\C$,
  and may be useful for subsequent work. 
Indeed Corollary \ref{corr-poly-eval} shows that $\SW_{x,w}(\zeta)$ can be evaluated at any 
  $\zeta\in B_{1-\delta}(\delta)$ up to a small error
  using traces of $x$,
  where $B_{1-\delta}(\delta)$ denotes the complex disc with center 
  $\delta$ and radius $1-\delta$.
We need $\delta\le 1/2$ so that $0\in B_{1-\delta}(\delta)$.
  
%can be evaluated for all $\zeta \in B_{1-\delta}(\delta)$, the complex disc with center $\delta$ and radius $1-\delta$, up to small error given access to traces of $x$.} %\blue{In \Cref{sec_compute_multiplicity_BAD}, we build upon \Cref{thm:uniqueness-info-theory} \red{and the algorithm of \Cref{sec_estimate_poly_values}}

\subsection{Evaluating $\SW_{x,w}(\zeta)$ for  $\zeta \in B_{1-\delta}(\delta)$ using traces of $x$} \label{sec_estimate_poly_values}

%\rnote{I guess what will go in this section is the efficient-algorithm-extension of the kind of thing that is in Section~0.3 of subword-counts2.pdf, right?  I think that this section, with the algorithm, should come a bit later --- we should at least state the structural results about polynomials (basically, what was called the HFL in that file) that motivates the algorithm and underlies its correctness, before we give the algorithm.  If we want we can prove the structural results about polynomials later but I think we should state them first. \blue{Sandip: does this layout work?}}
In the rest of this section we consider $\SW_{x,w}(\zeta)$ as a polynomial over complex numbers.
The main technical ingredient in the algorithm {\tt Multiplicity}$'_{\text{small-}\delta}$ is the following theorem, which relates~the subword polynomial to traces drawn from the deletion channel:

\begin{theorem}
    \label{thm:Taylor-subword-poly}
    Let $x, k$ and $w$ be as above. Then for all $\zeta \in \C$
    %\cnote{this holds for all $\zeta \in \C$ because it is an entire function, so the radius of convegence at any point is infinite.}\snote{do we need this restriction for this equality? I forgot about radii of convergence details years ago}, \rnote{I think this equality holds for all $\zeta \in \C$?  But the estimation (the corollary) will only be efficient for $\zeta$ such that $|{\frac {\zeta - \delta}{1-\delta}}| < 1$, right?} 
    we have
    \begin{equation*}
        %\label{eqn:subword-poly-eval}
        \SW_{x,w}(\zeta) = \frac{1}{(1-\delta)^k} \; \sum_{\mathclap{\substack{\alpha \in \Z_{\geq 0}^{k-1} \\ |\alpha| \leq n-k}}} \, \E_{\by \sim \Del_\delta(x)}\hspace{-0.04cm} \Big[ \#(w_0 \ast^{\alpha_1} w_1 \ast^{\alpha_2} w_2 \cdots w_{k-2} \ast^{\alpha_{k-1}} w_{k-1},\hspace{0.06cm} \by) \Big]\cdot  \left( \frac{\zeta-\delta}{1-\delta} \right)^{|\alpha|}.
    \end{equation*}
\end{theorem}

Before proving \Cref{thm:Taylor-subword-poly} in \Cref{sec:generalized} we use it to obtain
  the following corollary.

\begin{corollary}[Corollary of \Cref{thm:Taylor-subword-poly}] \label{corr-poly-eval}
    Let $x,k,w$ be as above, and let $\eps > 0$. Then, given access to traces $\by \sim \Del_\delta(x)$, there exists an algorithm which, given as input any $\zeta \in B_{1-\delta}(\delta)$, evaluates $\SW_{x,w}(\zeta)$  up to error $\pm \eps$ with success probability at least $1 - \tau$. The algorithm takes
    $$\left(\frac{n}{1-\delta}\right)^{O(k)}\cdot \frac{1}{\eps^2} \cdot \log \left(\frac{1}{\tau}\right)$$ many traces and running time.
\end{corollary}

Recall that $\SW_{x,w}(0) = \#(w,x)$.
 When $\delta \le 1/2$, the disc $B_{1-\delta}(\delta)$ contains the origin.  Therefore, setting $\eps = 1/3$ in Corollary~\ref{corr-poly-eval} directly implies an algorithm {\tt Multiplicity}$'_{\text{small-}\delta}$   that 
 % sets $\eps=1/3$ and 
 uses
  $((n / (1-\delta))^{O(k)} ) \cdot \log (1/\tau) = n^{O(k)} \cdot \log (1/\tau)$ traces 
   and running time to evaluate $\SW_{x,w}(0)$ up to an error of $\eps=1/3$,
  which succeeds with probability at least $1-\tau$. %$1-\tau/2$. 
It then rounds the result to the nearest integer to obtain $\SW_{x,w}(0)=\#(w,x)$ given that the latter
  is an integer. 
This finishes the proof of Theorem~\ref{thm:Multiplicity-small-delta-weak}.

\begin{proof}[Proof of Corollary \ref{corr-poly-eval}]
The algorithm simply draws $$s=\left(\frac{n}{1-\delta}\right)^{O(k)}\cdot \frac{1}{\eps^2} \cdot \log \left(\frac{1}{\tau}\right)$$ many traces $\by_1,
  \ldots,\by_s$ of $x$ and uses them to compute an empirical estimate $\tilde{E}_\alpha$ of
\begin{equation}\label{exp} E_\alpha:=\Ex_{\by \sim \Del_\delta(x)}\hspace{-0.04cm} \Big[ \#(w_0 \ast^{\alpha_1} w_1 \ast^{\alpha_2} w_2 \cdots w_{k-2} \ast^{\alpha_{k-1}} w_{k-1},\hspace{0.06cm} \by) \Big]\end{equation} 
for each $\alpha \in \Z_{\geq 0}^{k-1}$ with $|\alpha| \leq n-k$.
This is done by computing $\#(w_0 \ast^{\alpha_1} w_1 \ast^{\alpha_2} w_2 \cdots w_{k-2} \ast^{\alpha_{k-1}} w_{k-1},\hspace{0.06cm} \by_i)$
  for each $\alpha$ and $\by_i$ (in time polynomial in $n$)
  for each pair of $\alpha$ and $\by_i$),
  and then taking the average over $\by_1,\ldots,\by_s$ for each $\alpha$.
Given that the number of $\alpha$'s is at most $n^k$, the overall running time is 
  $s\cdot n^k\cdot \poly(n)$,
  as stated in Corollary \ref{corr-poly-eval}.
  
%computes an empirical estimate of each term in Equation \ref{eqn:subword-poly-eval} using $s = ((n / (1-\delta))^{O(k)}/\eps^2) \cdot \log (1/\tau)$ traces of $x$,\rnote{explain a bit more - is there a dynamic programming being done to compute $\#(w_0 \ast^{\alpha_1} w_1 \ast^{\alpha_2} w_2 \cdots w_{k-2} \ast^{\alpha_{k-1}} w_{k-1}, y)$ for each $y$, or something else?}
% and combines them to evaluate $\SW_{x,w}(\zeta)$.

Given that $\#(w_0 \ast^{\alpha_1} w_1 \ast^{\alpha_2} w_2 \cdots w_{k-2} \ast^{\alpha_{k-1}} w_{k-1},\hspace{0.06cm} \by)$ in (\ref{exp}) is between $0$ and $n$, 
  it follows from our choice of $s$, a Chernoff bound and a union bound, that with probability at least $1-\tau$,
  every empirical estimate $\tilde{E}_\alpha$ satisfies
  \begin{equation} \label{eq:one-estimate}
  |\tilde{E_\alpha}-E_\alpha|\le \eps \cdot \left(\frac{1-\delta}{n}\right)^k .
  \end{equation}
  %by a Chernoff bound and a union bound.
Using 
$$
\left|\frac{\zeta-\delta}{1-\delta}\right|\le 1
$$
when $\zeta\in B_{1-\delta}(\delta)$, we can use $\tilde{E}_\alpha$ to obtain an estimate of
  $\SW_{x,w}(\zeta)$:
$$
\frac{1}{(1-\delta)^k}\sum_{\alpha} \tilde{E}_\alpha\cdot 
\left( \frac{\zeta-\delta}{1-\delta} \right)^{|\alpha|}
$$
and the estimate is correct up to error 
$$
\frac{1}{(1-\delta)^k}\sum_{\alpha} |\tilde{E}_\alpha -E_\alpha|\le \eps,
$$
where the inequality holds by \Cref{eq:one-estimate} given that the number of $\alpha$'s is no more than $n^k$.
%Now we argue the correctness of this algorithm. Assume \Cref{thm:Taylor-subword-poly} holds. Each term\rnote{Sorry, I am being dense; what quantity exactly does ``term'' here refer to? (I changed ``$n/(1-\delta)^k$'' to ``$(n/(1-\delta))^k$.) \textbf{Sandip: } Each term refers to 1 term of the sum (i.e. for a fixed $\alpha$), including the $(1-\delta)^k$ factor. It should be $n / (1-\delta)^k$ and not $(n/(1-\delta))^k$. The $n^k$ comes from the fact that there are $n^{k-1}$ terms. I'll think of a better way to write this.}  (indexed by $\alpha$) in Equation \ref{eqn:subword-poly-eval} is bounded between $0$ and $(n/(1-\delta))^k$, and hence can be estimated up to error $\eps / \binom{n-1}{k-1}$ using $s$ samples, except with failure probability at most $\tau / \binom{n-1}{k-1}$, by a standard Chernoff bound argument. The number of terms is $\binom{n-k + k-1}{k-1} = \binom{n-1}{k-1}$. Hence, by union bound, except with probability $\tau$, the estimates of all terms are correct up to error $\eps / \binom{n-1}{k-1}$, and hence the estimate of $\SW_{x,w}(\zeta)$ is correct up to error $\eps$.
\end{proof}

%\subsection{Proof using Taylor expansion}

\subsection{Generalized deletion polynomial and the proof of \Cref{thm:Taylor-subword-poly}}\label{sec:generalized}

%Now we prove \Cref{thm:Taylor-subword-poly}. We will need to analyze the Taylor expansion of a certain multi-variate extension of $\SW_{x,w}(\cdot)$.
%\rnote{Probably we should add some exposition, here or soon hereafter, to justify/explain why we are doing this in a more general way than we need for the theorem. }
In this subsection we prove \Cref{thm:Taylor-subword-poly}. 
We first introduce a more general class of polynomials, the $(x,f)$-deletion-channel polynomials 
  (see  \Cref{def:deletionpoly}), of which $\SW_{x,w}$ is a special case.
We then prove an extension of \Cref{thm:Taylor-subword-poly} (see \Cref{lemma:Taylor-of-del-poly})
  which applies to 
  every $(x,f)$-deletion channel polynomial; \Cref{thm:Taylor-subword-poly} follows as a direct corollary.
 %    and prove an equation 
While we don't need the full generality of \Cref{lemma:Taylor-of-del-poly} to 
  prove \Cref{thm:Taylor-subword-poly},
% We first define a more general class of polynomials of which $\SW_{x,w}$ is a special case.  This 
  working with this new class of polynomials makes our proofs cleaner.
We also believe that \Cref{lemma:Taylor-of-del-poly} in the general form
  may be useful for subsequent analysis.

The following notation will be convenient for us. Given vectors $\gamma \in \Z_{\ge 0}^k$ and $\xi \in \C^k$, and a function $P(z_1,\ldots,z_k)$ from $\C^k$ to $\C$, we define 
\[
  \qquad \xi^\gamma = \xi_1^{\gamma_1} \cdots \xi_k^{\gamma_k} \qquad\text{and}\qquad D^{\gamma} P = \frac{\partial^{\abs{\gamma}} P}{\partial z_1^{\gamma_1} \cdots \partial z_k^{\gamma_k}} .
\]
Recall that $\gamma! = \gamma_1! \cdots \gamma_k!$ and $\abs{\gamma} = \gamma_1 + \cdots + \gamma_k$.
%We sometimes write $\vec{v}$ to emphasize that $v$ is a vector.
\noindent For $v \in \C$, we will denote the vector $(v, v, \cdots, v) \in \C^k$ by $\vec{v}$, where the dimension $k$ will be clear from context.

We define the class of $(x,f)$-deletion-channel polynomials:
  \begin{definition}\label{def:deletionpoly}
    Given $f\colon \zo^k \to \C$ and a string $x \in \zo^n$, the {\em $(x,f)$-deletion-channel polynomial} $P_{x,f}\colon \C^k \to \C$ is defined by
    \[
      P_{x,f}(\xi) := \sum_{\substack{\gamma \in \Z_{\geq 0}^k\\ \abs{\gamma} \le n - k}} f(x_{\gamma_1}, x_{\gamma_1 + \gamma_2 + 1}, \ldots, x_{\gamma_1 + \dots + \gamma_k + (k-1)})\cdot \xi^\gamma .
    \]
  \end{definition}

%%  Later we will take $f\colon \zo^k \to \zo$ to be the indicator function $f(x) := \Indicator{(x_1,\ldots,x_k) = w}$ for every $w \in \zo^k$.  Note that 
%%  \[
%%    P_{x,f}(1,0, \ldots, 0) = \sum_{i=0}^{n-k} \Indicator{(x_i, x_{i+1}, \ldots, x_{i+k-1}) = w} = \#(w,x) ,
%%  \]
%%  which is the quantity we want to compute.  %We call $P_{x,f}$ the ``$(f,x)$-deletion-channel polynomial'' because of 
We call $P_{x,f}$ the $(x,f)$-deletion-channel polynomial because by choosing
   $k=1$ and $f \colon \zo \to \zo$ to be the $1$-bit identity function $\id(x) = x$, we have that
  \[
    P_{x,\id}(\xi) = \sum_{i=0}^{n-1} x_i \xi^i
  \]
  is the deletion-channel polynomial defined in~\cite{DOS17}.

  The next theorem shows that under a change of variables, the coefficients of $P_{x,f}$ with respect to the new variables can be expressed in terms of the expectation of $f$ over traces of $x$ drawn from the deletion channel. We state it and then show that \Cref{thm:Taylor-subword-poly} follows as a direct corollary. %This allows us to approximate $P_{x,f}$ over $B_{1-\delta}(\delta)$ using the traces $\by \sim \Del_\delta(x)$.

  \begin{theorem} \label{lemma:Taylor-of-del-poly}
    For any $\xi \in \C^k$, we have
    \[
      P_{x,f}(\xi) = \frac{1}{(1-\delta)^k} \; \sum_{\substack{\beta \in \Z_{\ge 0}^k\\ \abs{\beta} \le n - k}} \E_{\by \sim \Del_\delta(x)} \Big[ f(\by_{\beta_1}, \ldots, \by_{\beta_1 + \dots + \beta_k + k - 1}) \Big] \cdot \left( \frac{\xi - \vec{\delta}}{1-\delta} \right)^{\beta}.
    \]
  \end{theorem}

%  Given \Cref{lemma:Taylor-of-del-poly}, Theorem~\ref{thm:Taylor-subword-poly} is now immediate.
%\subsubsection{Special case: the subword polynomial}

  \begin{proof}[Proof of \Cref{thm:Taylor-subword-poly} assuming \Cref{lemma:Taylor-of-del-poly}]
    Given $x \in \zo^n$ and $w \in \zo^k$ for some $k \in [n]$ as in the statement of \Cref{thm:Taylor-subword-poly},
    we take $f : \zo^k \to \zo$ to be the indicator function of $w$:
    \[
        f(b_1, b_2, \ldots, b_k) = \Indicator{(b_1, b_2, \ldots, b_k) = w}.
    \]
Using this $f$ we get the following connection between $\SW_{x,w}(\zeta)$ and $P_{x,f}(1,\zeta,
  \zeta,\ldots,\zeta)$:
    \begin{align*}
      \SW_{x,w}(\zeta) &= \sum_{\mathclap{\substack{\alpha \in \Z_{\ge 0}^{k-1} \\ |\alpha| \leq n-k}}} \, \#\left(w_0 \ast^{\alpha_1} w_1 \ast^{\alpha_2} w_2 \cdots w_{k-2} \ast^{\alpha_{k-1}} w_{k-1}, \hspace{0.06cm}x\right) \, \cdot \zeta^{|\alpha|}\\
                       &= \sum_{\mathclap{\substack{\alpha \in \Z_{\ge 0}^{k-1} \\ |\alpha| \leq n-k}}} \sum_{i=0}^{n-k-|\alpha|} \, f(x_i, x_{i+\alpha_1+1}, x_{i + \alpha_1 + \alpha_2+2}, \ldots, x_{i + |\alpha| + k-1}) \, 1^i \, \zeta^{|\alpha|} 
        = P_{x,f}(1, \zeta, \zeta, \cdots, \zeta)\end{align*}
Applying \Cref{lemma:Taylor-of-del-poly} on $P_{x,f}(1, \zeta, \zeta, \ldots, \zeta)$, we have
\begin{align*}
        \SW_{x,w}(\zeta)&= \frac{1}{(1-\delta)^k} \; \sum_{\mathclap{\substack{\alpha \in \Z_{\geq 0}^{k-1} \\ |\alpha| \leq n-k}}} \sum_{i=0}^{n-k-|\alpha|} \E_{\by \sim \Del_\delta(x)}\Big [ f(\by_i, \by_{i+\alpha_1+1}, \ldots, \by_{i + |\alpha| + k-1})\Big ]
        \cdot \left( \frac{\zeta - \delta}{1-\delta} \right)^{|\alpha|}\\
        &= \frac{1}{(1-\delta)^k} \; \sum_{\mathclap{\substack{\alpha \in \Z_{\geq 0}^{k-1} \\ |\alpha| \leq n-k}}} \, \E_{\by \sim \Del_\delta(x)} \Big[ \#(w_0 \ast^{\alpha_1} w_1 \ast^{\alpha_2} w_2 \cdots w_{k-2} \ast^{\alpha_{k-1}} w_{k-1}, \hspace{0.06cm}\by) \Big] \cdot \left( \frac{\zeta-\delta}{1-\delta} \right)^{|\alpha|} 
    \end{align*}
  where the last step follows by linearity of expectation. This concludes the proof of \Cref{thm:Taylor-subword-poly}.
  \end{proof}

  We now prove \Cref{lemma:Taylor-of-del-poly}.  The high-level idea is to relate the expectation of $f$ over traces of $x$ drawn from the deletion channel to partial 
  derivatives of polynomial $P_{x,f}$ at $\smash{\vec{\delta}}$, and then apply Taylor's expansion to $P_{x,f}$ at the point $\smash{\vec{\delta}}$.

  \begin{claim} \label{claim:derivatives}
    Let $\beta \in \Z_{\ge 0}^k$ with  $\abs{\beta} \le n - k$.
    We have
    \[
      \E_{\by \sim \Del_\delta(x)} \Big[ f(\by_{\beta_1}, \ldots, \by_{\beta_1 + \dots + \beta_k + (k-1)}) \Big]
      = (1 - \delta)^k \cdot \frac{(1-\delta)^{\abs{\beta}}}{\beta!} \cdot D^\beta P_{x,f}(\vec{\delta}\hspace{0.04cm}) .
    \]
  \end{claim}

  To get some intuition, consider the special case of $k=1$ (so $P_{x,f}$ is univariate) and $f = \id$. Then it is straightforward to verify that
  \[
    \E_{\by \sim \Del_\delta(x)}\big[\by_0\big] = (1-\delta) \sum_{i=0}^{n-1} x_i \delta^i = (1-\delta) \cdot P_{x,\id}(\delta) ,
  \]
  and
  \[
    \E_{\by \sim \Del_\delta(x)}\big[\by_1\big]
    = (1-\delta) \sum_{i=1}^{n-1} x_i \binom{i}{1} (1-\delta) \delta^{i-1}
    = (1-\delta)^2 \sum_{i=1}^{n-1} x_i i \delta^{i-1}
    = (1-\delta)^2 \cdot D^1 P_{x,\id}(\delta). %\frac{dP_{x,\id}}{d \xi}(\delta).
  \]

  \begin{proof}[Proof of Claim~\ref{claim:derivatives}]
    For a fixed $\gamma \in \Z_{\geq 0}^k$ with $|\gamma| \leq n-k$, we write
    \[
      \gamma \to \beta,
      \quad \quad \text{or equivalently}
      \quad \quad  (\gamma_1, \gamma_2, \ldots, \gamma_k) \to (\beta_1, \beta_2, \ldots, \beta_k),
    \]
    to denote the event that %when
    the $(\gamma_1, \gamma_1 + \gamma_2 + 1, \ldots, \gamma_1 + \dots + \gamma_k + (k-1))$ positions of $x$ become the $(\beta_1, \beta_1 + \beta_2 + 1, \ldots, \beta_1 + \dots + \beta_k + (k-1))$ positions of $\by \sim \Del_\delta(x)$ respectively.  For this to occur, each bit $x_{\gamma_i}$ must be present in $\by$, which happens with probability $(1-\delta)^k$.  Further, for each $x_{\gamma_i}$ to become $\by_{\beta_i}$, exactly $\beta_i$ out of the $\gamma_i$ bits between (and including) positions $\gamma_1 + \cdots + \gamma_{i-1} + i$ and $\gamma_1 + \cdots + \gamma_i + (i-1)$ of $x$ must be retained. So, the probability of this event is
    \begin{align}
      \Pr[\gamma \to \beta]
      &= %(1-\delta)^k \cdot \binom{\gamma_1}{\beta_1} (1-\delta)^{\beta_1} \delta^{\gamma_1-\beta_1} \cdot \binom{\gamma_2}{\beta_2} (1-\delta)^{\beta_2} \delta^{\beta_2 -\gamma_2} \cdots \binom{\gamma_k}{\beta_k} (1-\delta)^{\beta_k} \delta^{\beta_k -\gamma_k} \\
      %&=
      (1-\delta)^k \prod_{i=1}^k \binom{\gamma_i}{\beta_i} (1-\delta)^{\beta_i} \delta^{\gamma_i - \beta_i} \nonumber\\
      &= (1-\delta)^k \prod_{i=1}^k \frac{\gamma_i (\gamma_i-1) \cdots (\gamma_i - \beta_i + 1)}{\beta_i!} (1-\delta)^{\beta_i} \delta^{\gamma_i - \beta_i}  \nonumber \\
      &= (1-\delta)^k \cdot \left( \prod_{i=1}^k \frac{(1-\delta)^{\beta_i}}{\beta_i!} \right) \cdot \prod_{i=1}^k \left(\gamma_i (\gamma_i-1) \cdots (\gamma_i - \beta_i + 1) \cdot \delta^{\gamma_i - \beta_i}\right) \nonumber \\
      &= (1-\delta)^k \cdot \frac{(1-\delta)^{\abs{\beta}}}{\beta!} \cdot \prod_{i=1}^k \frac{d^{\beta_i}}{d \delta} \delta^{\gamma_i}. \label{eq:expr}
    \end{align}    
    \begin{comment}
    The factor $(1-\delta)^k$ is because the bits $x_{\gamma_i}$ must be present in $\by$. Then to retain each $x_{\gamma_i}$ so that it becomes $y_{\beta_i}$, we can retain $\beta_i$ out of the $\gamma_i$ bits between positions $\gamma_1 + \cdots + \gamma_{i-1} + i$ and $\gamma_1 + \cdots + \gamma_i + (i-1)$ of $x$.
    Now,
    \[
      \E_{\by \sim\Del_\delta(x)} [ f(\by_{\beta_1}, \ldots, \by_{\beta_1 + \dots + \beta_k + (k-1)}) ]
      = \sum_{\abs{\gamma}\le n-k} f(x_{\gamma_1}, \ldots, x_{\gamma_1 + \cdots + \gamma_k + (k-1)}) \Pr[\gamma \to \beta] .
    \]
    For a fixed $\gamma$, 
    \begin{align*}
      \Pr[\gamma \to \beta]
      &= (1-\delta)^k \prod_{i=1}^k \binom{\gamma_i}{\beta_i} (1-\delta)^{\beta_i} \delta^{\gamma_i - \beta_i} \\
      &= (1-\delta)^k \prod_{i=1}^k \frac{\gamma_i (\gamma_i-1) \cdots (\gamma_i - \beta_i + 1)}{\beta_i!} (1-\delta)^{\beta_i} \gamma^{\gamma_i - \beta_i} \\
      &= (1-\delta)^k \cdot \left( \prod_{i=1}^k \frac{(1-\delta)^{\beta_i}}{\beta_i!} \right) \cdot \prod_{i=1}^k \gamma_i (\gamma_i-1) \cdots (\gamma_i - \beta_i + 1) \delta^{\gamma_i - \beta_i} \\
      &= (1-\delta)^k \cdot \frac{(1-\delta)^{\abs{\beta}}}{\beta!} \cdot \prod_{i=1}^k \frac{d^{\beta_i}}{d \delta} \delta^{\gamma_i} .
    \end{align*}
    \end{comment}
As a result, we have that
    \begin{align*}
      &\hspace{-0.6cm}\E_{\by \sim \Del_\delta(x)} \Bigl[ f(\by_{\beta_1}, \ldots, \by_{\beta_1 + \dots + \beta_k + (k-1)}) \Bigr] \\[1ex]
      &= \sum_{\abs{\gamma}\le n-k} f(x_{\gamma_1}, \ldots, x_{\gamma_1 + \cdots + \gamma_k + (k-1)})\cdot \Pr[\gamma \to \beta] \\
      &= (1-\delta)^k \cdot \frac{(1-\delta)^{\abs{\beta}}}{\beta!} \sum_{\abs{\gamma}\le n-k} f(x_{\gamma_1}, \ldots, x_{\gamma_1 + \cdots + \gamma_k + (k-1)})\cdot  \prod_{i=1}^k \frac{d^{\beta_i}}{d \delta} \delta^{\gamma_i} \tag{\Cref{eq:expr} } \\
      &= (1-\delta)^k \cdot \frac{(1-\delta)^{\abs{\beta}}}{\beta!} \cdot D^\beta P_{x,f}(\vec{\delta}\hspace{0.04cm}).  
    \end{align*}
This finishes the proof of Claim~\ref{claim:derivatives}.
  \end{proof}

  \begin{proof}[Proof of \Cref{lemma:Taylor-of-del-poly}]
{Since $P_{x,f}$ is a polynomial of degree at most $n-k$,}
    applying Taylor's expansion to $P_{x,f}$ at the point $\smash{\vec{\delta}}$ and using Claim~\ref{claim:derivatives}, we get that
  \begin{align*}
    (1-\delta)^k\cdot  P_{x,f}(\xi)
    &= (1-\delta)^k \sum_{\abs{\beta} \le n - k} \frac{D^{\beta}P_{x,f}(\vec{\delta}\hspace{0.04cm})}{\beta!} \cdot (\xi - \vec{\delta})^{\beta} \\
    &= \sum_{\abs{\beta} \le n - k} \E_{\by \sim \Del_\delta(x)} \Big[ f(\by_{\beta_1}, \ldots, \by_{\beta_1 + \dots + \beta_k + k - 1}) \Big]\cdot  \left( \frac{\xi - \vec{\delta}}{1-\delta} \right)^{\beta}. \qedhere
  \end{align*}
  \end{proof}

\section{{\tt Multiplicity}$'_{\text{large-}\delta}$:  An algorithm for deletion rate $\delta < 1$} \label{sec:large-delta}

In this section we prove a weaker version of \Cref{thm:Multiplicity}, giving an algorithm that works for any
  deletion rate $\delta < 1$ but has quasipolynomial running time and sample complexity when $k \approx \log n$ (as will be the case in our ultimate application):

\begin{theorem} \label{thm:Multiplicity-large-delta-weak} 
Let $0 < \tau,\delta < 1.$ There is an algorithm 
{\tt Multiplicity}$'_{\text{large-}\delta}$ which takes as input a string $w \in \zo^k$ and access to independent traces of an unknown source string $x \in \zo^n$.
{\tt Multiplicity}$'_{\text{large-}\delta}$ runs in  $\left(\frac{n^{1/(1-\delta)}}{1-\delta}\right)^{O(k)} \log \left(\frac{1}{\tau}\right)$ %$\big(\frac{n}{1-\delta}\big)^{O(k/(1-\delta))} \log \frac{1}{\tau}$
time and uses $\left(\frac{n^{1/(1-\delta)}}{1-\delta}\right)^{O(k)} \log \left(\frac{1}{\tau}\right)$  %$\big(\frac{n}{1-\delta}\big)^{O(k/(1-\delta))} \log \frac{1}{\tau}$
many traces from $\Del_\delta(x)$, and has the following property:  For any unknown source string $x \in \zo^n$, with probability at least $1-\tau$ %$1-\tau/2$
the output of {\tt Multiplicity}$'_{\text{large-}\delta}$ is $\#(w,x)$, the multiplicity of $w$ in $\subword(x,k)$ \emph{(}equivalently, the value $\SW_{x,w}(0)$\emph{)}.
\end{theorem}

Looking ahead, in \Cref{sec:strong-version} we will build on the proof of 
\Cref{thm:Multiplicity-large-delta-weak} to give a stronger version that has polynomial running time and sample complexity when $k \approx \log n$.

The following result is central to our analysis.  Informally, it says that if $q$ is a polynomial with ``not-too-large'' coefficients and a constant term which is bounded away from $\SW_{x,w}(0)$ by at least $1/2$, %\blue{at least $1/2$ away from $\SW_{x,w}(0)$} 
then $q$ must ``differ noticeably'' from $\SW_{x,w}$ over a particular interval. (Looking ahead, for our purposes it is crucially important that this interval corresponds to a range of deletion probabilities for which it is easy to estimate the polynomial's value given access to traces drawn from $\Del_\delta(x)$.)

%\blue{For technical reasons, we require the slightly stronger condition that the constant terms differ by at least $1/2$.}\snote{Added this sentence. We need the current form because the LP may not return an integer. Also I changed $|q_\ell| \leq n^k$ to $0 \leq q_\ell \leq n^k$.}\ignore{\rnote{This was ``the values $\SW_{x,w}(\zeta)$ are sufficient to uniquely identify $\#(w,x)$, the multiplicity of $w$ in $\subword(x,k)$''}}

\begin{theorem} \label{thm:uniqueness-info-theory}
    Fix strings $x \in \zo^n, w \in \zo^k$ for some $k \in [n]$. Let $q(z) = \sum_{\ell=0}^{n-k} q_\ell \, z^\ell$ be any polynomial such that $|\SW_{x,w}(0) - q(0)| \geq 1/2$, %$q(0)$ is an integer, $q(0) \neq \#(w,x)$
    and $0 \leq q_\ell \leq n^k$ %$|q_i| \leq n^k$ 
    for all $\ell \in \{0,1,\cdots,n-k\}$. 
    %\rnote{$|q(i)|$ or $q_i$?}%whose multiplicity in $\subword(x,k)$ is different from that in $\subword(x',k)$.
    Then 
\begin{equation}\label{hehe}
        \sup_{\zeta \in [\delta,(\delta+1)/2]} \big\lvert \SW_{x,w}(\zeta) - q(\zeta) \big\rvert \geq n^{-O(k/(1-\delta))},\quad\text{for any $\delta \in (0,1)$.}
\end{equation}%    for some $\zeta \in [\delta,1]$.
\end{theorem}

\begin{comment}
\begin{theorem} \label{thm:uniqueness-info-theory}
    Let $x, x' \in \zo^n$ be distinct strings such that $\#(w,x) \neq \#(w,x')$ for some $k \in [n], w \in \zo^k$. %whose multiplicity in $\subword(x,k)$ is different from that in $\subword(x',k)$.
    Then for any $\delta \in (0,1)$,
    \[
        \sup_{\zeta \in [\delta,(\delta+1)/2]} \lvert \SW_{x,w}(\zeta) - \SW_{x',w}(\zeta) \rvert \geq n^{-O(k/(1-\delta))}.
    \]%    for some $\zeta \in [\delta,1]$.
\end{theorem}
\end{comment}

%Later, we will build upon this result to design an algorithm to recover $\#(w,x)$, given access to traces of $x$, that is efficient in terms of sample complexity but computationally inefficient. Finally, we will present a computationally efficient algorithm for the same task based on linear programming.\snote{Add references to relevant algorithms.} \Cref{thm:uniqueness-info-theory} can be obtained from the following more general theorem:

\Cref{thm:uniqueness-info-theory} is an easy consequence of the following more general theorem:

\begin{theorem} \label{thm:complex} 
Let $1 \leq n \leq m$.  Let $p(z) = \sum_{\ell=0}^n p_\ell\, z^\ell$ be a polynomial of degree at most $n$ with real coefficients such that $|p_0| \geq 1/2$, and $|p_\ell|\le m$ for all $\ell$. Then we have 
\begin{equation}\label{final00}
\sup_{\zeta \in [\delta,(\delta+1)/2]} \big|p(\zeta)\big| \geq m^{-O\left(1/(1-\delta)\right)},\quad
\text{for any $\delta \in (0,1)$.}
\end{equation}
%for some positive constant $c(\delta)>0$. 
\end{theorem}

To obtain \Cref{thm:uniqueness-info-theory}  from \Cref{thm:complex}, set $p = \SW_{x,w} - q$. By the condition of \Cref{thm:uniqueness-info-theory} we have that $|p_0| = |\SW_{x,w}(0) - q_0| \geq 1/2$. Writing $(\SW_{x,w})_\ell$ for the degree-$\ell$ coefficient of $\SW_{x,w}$, from the discussion following \Cref{def:subword-polynomial} it is immediate that $0 \leq (\SW_{x,w})_\ell \leq {n \choose k} \leq n^k$, and hence $|p_\ell| = |(\SW_{x,w})_\ell - q_\ell| \leq n^k$. %$|p_\ell| \leq (\SW_{x,w})_\ell + |q_\ell| \leq 2n^k$.
Thus we can invoke \Cref{thm:complex} with $m=n^k$ to obtain \Cref{thm:uniqueness-info-theory}.

In \Cref{sec_compute_multiplicity_EFFICIENT} we present and analyze the algorithm {\tt Multiplicity}$'_{\text{large-}\delta}$ (which is based on 
  linear programming) and prove \Cref{thm:Multiplicity-large-delta-weak} assuming \Cref{thm:uniqueness-info-theory}.
The proof of \Cref{thm:complex}, which is based on complex analysis, is given in \Cref{sec:proof:thm_complex}.

\begin{comment}
To see that \Cref{thm:uniqueness-info-theory} follows from \Cref{thm:complex}, set $p = \SW_{x,w} - \SW_{x',w}$. Then $|p(0)| = |\#(w,x) - \#(w,x')| \geq 1$. Moreover, we have $|p_i| \leq n^k$ for all $i \in \{0,1,\cdots,n\}$ because the number of terms $\alpha \in \Z_{\geq 0}^{k-1}$ satisfying $|\alpha| = i$ is at most $n^{k-1}$, and the quantity in each term is bounded between $-n$ and $n$. So, we can invoke \Cref{thm:complex} with $m=n^k$ to obtain \Cref{thm:uniqueness-info-theory}.
\end{comment}

\begin{comment}
\begin{lemma}
    Fix any deletion rate $\delta > 0$. For any position $j \in [0,n-k-1]$ and vector $\beta \in \Z_{\geq 0}^{k-1}$ with $|\beta|\leq n-j$, we have
    \[
        \Pr_{\by \sim \Del_\delta(x)}\left[(\by_j, \by_{j+\beta_1+1}, \cdots, \by_{j + \beta_1 + \cdots + \beta_{k-1} + (k-1)}) = w\right] = ???.
    \]
\end{lemma}
\end{comment}

\subsection{Proof of \Cref{thm:Multiplicity-large-delta-weak} assuming \Cref{thm:uniqueness-info-theory}} \label{sec_compute_multiplicity_EFFICIENT}

\subsubsection{Estimating $\SW_{x,w}(\delta')$ %(\zeta)$
for %$\zeta
$\delta' \geq \delta$}
The following easy lemma gives an unbiased estimator for $\SW_{x,w}(\delta')$, for all $\delta' \geq \delta$, %$\#(w,x)$
given traces from $\Del_{\delta}(x)$. 
\begin{lemma}~\label{lem:estimator-polynomial}
Let $x \in \zo^n$, $w \in \zo^k$ and let $\epsilon>0$. Then, given traces $\by \sim \Del_{\delta}(x)$, there 
exists an algorithm, which for any $\delta' \in [\delta,1]$, %$\delta' \ge \delta$,
evaluates $\SW_{x,w}(\delta')$ up to error $\pm \epsilon$ with success probability at least $1-\tau$. The algorithm takes 
\[
n^{O(1)} \cdot \bigg( \frac{1}{1-\delta'}\bigg)^{O(k)} \cdot \frac{1}{\epsilon^2} \cdot \log \bigg( \frac{1}{\tau} \bigg)
\]
many traces and running time. 
\end{lemma}
\begin{proof}
First of all, observe that given $\by \sim \Del_{\delta}(x)$, we can sample $\by \sim \Del_{\delta'}(x)$ for any $\delta' \ge \delta$ with no overhead. 
%Thus, from now, we focus on $\delta'=\delta$. 
%it suffices to prove the above lemma just for $\delta'=\delta$.\snote{I think it's better to keep the proof in terms of general $\delta' \geq \delta$. I feel like the current proof could give the impression that the sample complexity is proportional to $1/(1-\delta)^k$, instead of $1/(1-\delta')^k$, and also this reduction doesn't seem to simplify the proof much.}
% This is because, 
Next, observe that the expected number of $w$ in a randomly trace $\by \sim \Del_{\delta'}(x)$ is given by 
\[
  \Ex_{\by \sim \Del_{\delta'}(x)}[\#(w,\by)] =
  \sum_{\mathclap{\substack{\alpha \in \Z_{\geq 0}^{k-1} \\ |\alpha| \leq n-k}}} \, \#\left(w_0 \ast^{\alpha_1} w_1 \ast^{\alpha_2} w_2 \cdots w_{k-2} \ast^{\alpha_{k-1}} w_{k-1},\hspace{0.06cm} x\right) \, \cdot {\delta'}^{|\alpha|} \cdot (1-\delta')^k. 
 \]
This follows from the fact that every occurrence of $w$ as a subword of trace $y$  can be uniquely identified with a  subsequence $(i_1\le \ldots \le i_k)$ such that  (i) 
$x_{i_1}=w_1 \wedge \ldots \wedge x_{i_k}=w_k$. (ii) positions $i_1, \ldots, i_k$ are not deleted in $\by$.  (iii) every position in $[i_1, \ldots, i_k] \setminus \{i_1, \ldots, i_k\}$ is deleted in the trace $\by$.  However, by \Cref{def:subword-polynomial}, it follows that 
\begin{align} 
  \Ex_{\by \sim \Del_{\delta'}(x)}[\#(w,\by)]
&= \SW_{x,w}(\delta') \cdot (1-\delta')^k. 
\label{eq:relation-emp-trace-sw}
\end{align}
Now for any $\by \sim \Del_{\delta'}(x)$, 
$\#(w,\by)$ is an integer between $0$ and $n$. Thus, the standard empirical estimator will use 
\[
n^{O(1)} \cdot \bigg( \frac{1}{1-\delta'}\bigg)^{O(k)} \cdot \frac{1}{\epsilon^2} \cdot \log \bigg( \frac{1}{\tau} \bigg)
\]
many traces and running time and returns an estimate of 
$\Ex_{\by \sim \Del_{\delta'}(x)}[\#(w,\by)]$ up to $\pm \epsilon \cdot (1-\delta')^{k}$. Using \eqref{eq:relation-emp-trace-sw}, we get the claim. 
%pairofcoordinatesxi =b1,xj =b2 withi<j,andifxi =b1 andxj =b2 forsomei<j,thenthe probability that those two coordinates both are present in y is (1 − δ)2 and the probability that
\end{proof}

\subsubsection{The {{\tt Multiplicity}$'_{\text{large-}\delta}$ algorithm and its analysis}}

\begin{figure}[h]
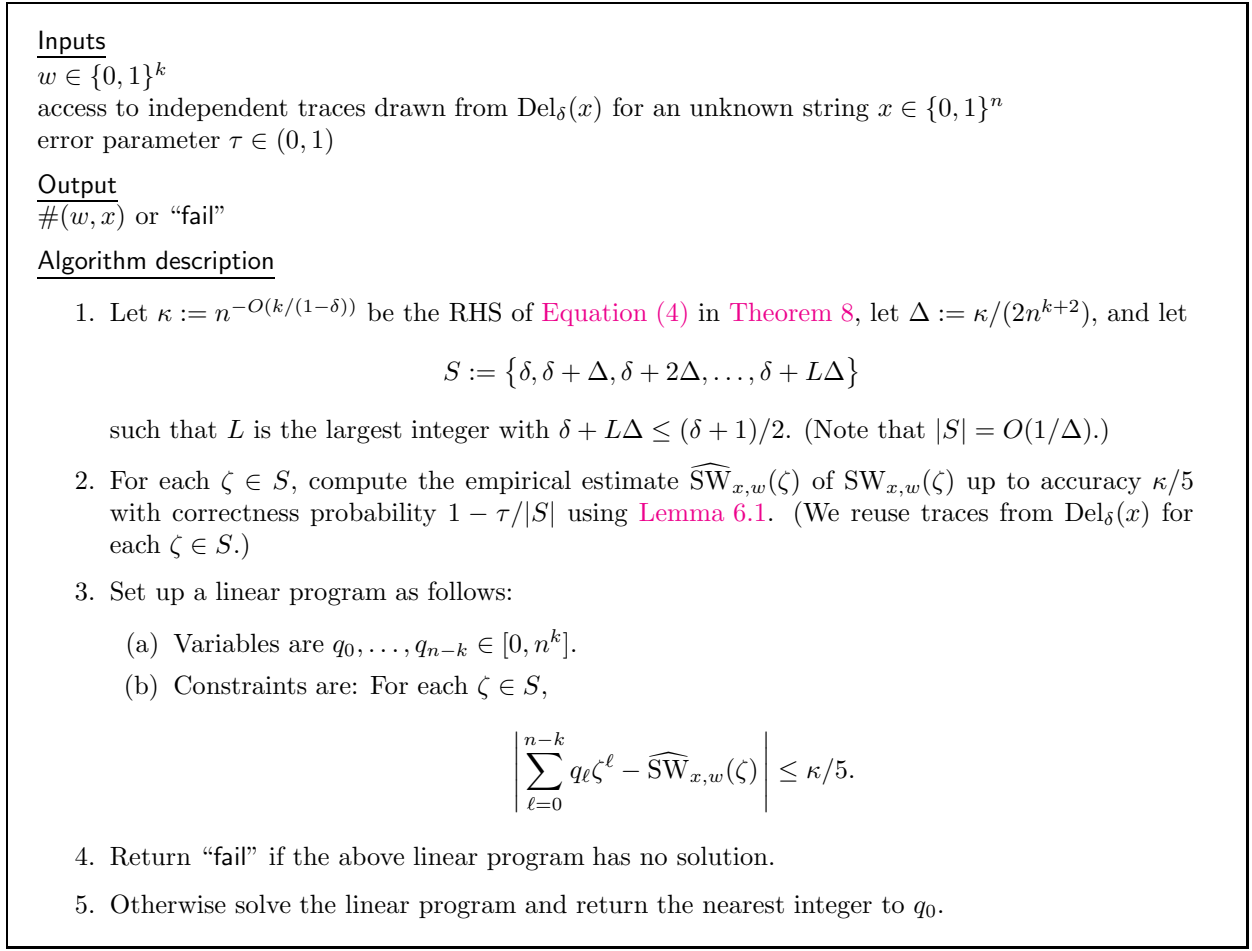

\hrule
\vline
\begin{minipage}[t]{0.98\linewidth}
\vspace{10 pt}
\begin{center}
\begin{minipage}[h]{0.95\linewidth}
{\small
\underline{\textsf{Inputs}}
~\\
$w \in \{0,1\}^k$ \\ 
access to independent traces drawn from $\Del_\delta(x)$ for an unknown string $x\in \{0,1\}^n$  \\
error parameter $\tau\in (0,1)$ 
% noise parameter $\delta \in (0,1)$. 
 
\vspace{5 pt}
\underline{\textsf{Output}}
~\\ 
$\# (w,x)$ or ``\textsf{fail}''

\vspace{5pt}
%\underline{\textsf{Parameters}}

%$R$ &:=& Number of points where we evaluate each test function is evaluated\\

%\vspace{5 pt}
\underline{\textsf{Algorithm description}}
\begin{enumerate}
  \item Let $\kappa := n^{-O(k/(1-\delta))}$ be the RHS of \Cref{hehe}
    in \Cref{thm:uniqueness-info-theory}, let $\Delta := \kappa/(2n^{k+2})$,
  and let 
%$$\gamma=\frac{\kappa}{2n^{k+1}} \quad\text{and}\quad S = \big\{\delta,\delta(1+\gamma),
%\delta(1+\gamma)^2,\ldots,\delta(1+\gamma)^L\big\}$$
%where $L$ is the smallest integer with $\delta(1+\gamma)^L\ge (\delta+1)/2$; so
  \[ S :=\big\{\delta, \delta + \Delta, \delta + 2\Delta, \ldots, \delta+L\Delta \bigr\} \]
  such that $L$ is the largest integer with $\delta+L\Delta\le (\delta+1)/2$.
(Note that $|S|= O(1/\Delta)$.)
\item For each $\zeta \in S$, compute the empirical estimate $\hat{\SW}_{x,w}(\zeta)$ of $\SW_{x,w}(\zeta)$
  up to accuracy $\kappa/5$ with correctness probability $1-\tau/|S|$ %$1- \tau/(2|S|)$ %\cnote{Why is this $1-\tau\gamma/2$?} 
  using \Cref{lem:estimator-polynomial}. %\cnote{It seems that Cref mistreats the Corollary as Lemma.  I'm replacing it with ref}.  
  (We reuse traces from $\Del_\delta(x)$ for each $\zeta  \in S$.)
\item Set up a linear %\xnote{I thought this is a linear program; we just add two 
  %linear constraints for each inequality.} 
  program as follows: 
\begin{enumerate} \item Variables are 
$q_0, \ldots, q_{n-k} \in [0,n^k]$. 
\item Constraints  are: For each $\zeta \in S$, $$\left|\hspace{0.04cm}\sum_{\ell=0}^{n-k} 
q_{\ell} \zeta^{\ell} - \hat{\SW}_{x,w}(\zeta)\hspace{0.04cm}\right| \le \kappa/5.$$
\end{enumerate}
\item Return ``\textsf{fail}'' if the above linear program has no solution.
\item 
Otherwise solve the linear program and return the nearest integer to $q_0$. 
\end{enumerate}

\vspace{5 pt}
}
\end{minipage}
\end{center}

\end{minipage}
\hfill \vline
\hrule
\begin{center}
 \caption{Description of the algorithm {\tt Multiplicity}$'_{\text{large-}\delta}$.}
\label{fig:tlin-1}
\end{center}
\end{figure}

We present the algorithm {{\tt Multiplicity}$'_{\text{large-}\delta}$} in \Cref{fig:tlin-1}. %\rnote{It seems like the Cref command is botching this reference - this should say "Figure 1".  I'm not sure how to fix it.}
For its correctness we first observe that with probability at least $1-\tau$, %$1-\tau/2$,
we have that 
\[
\text{for 
  every $\zeta\in S$,} \quad \left|\hat{\SW}_{x,w}(\zeta) -{\SW}_{x,w}(\zeta)\right| \le \kappa/5.
  \]
  %Next, observe that the \red{linear} program in lines 3(a) and 3(b) is feasible. 
We finish the proof by showing that when this happens, 
  the linear program in lines 3(a) and 3(b) is feasible,
  and furthermore, $|q_0-\SW_{x,w}(0)|<1/2$ in any feasible solution $(q_0,\ldots,q_{n-k})$ (when this happens, the closest integer to $q_0$
    is exactly $\SW_{x,w}(0)$).
  
  To see that the linear program is feasible, we let $p_0,\ldots,p_{n-k}$ denote the coefficients of ${\SW}_{x,w}$, so
  %this, let us define $p_0, \ldots, p_{n-k}$ so that as a formal polynomial, 
  ${\SW}_{x,w}(\zeta)=\sum_{\ell=0}^{{n-k}} p_\ell\hspace{0.03cm} \zeta^\ell$.
%  We can now set 
% $q_i = p_i$ for all $0 \le i \le n-k$. Observe that  
% $p_0 = \#(w,x)$. 
 %This is because we can set $q_0, \ldots, q_{n-k}$ such that as a formal polynomial, $\sum_{\ell} q_\ell \zeta^\ell = {\SW}_{x,w}(\zeta)$.  
 From the discussion after \Cref{thm:complex}, every $p_\ell$ lies between $0$ and 
   $n^k$.
As a result, $p_0,\ldots,p_{n-k}$ is a feasible solution to the linear program 
  because for every $\zeta\in S$,  
$$\left|\hspace{0.04cm}\sum_{\ell=0}^{n-k} 
p_{\ell} \zeta^{\ell} - \hat{\SW}_{x,w}(\zeta)\hspace{0.04cm}\right|= 
\left|\hspace{0.04cm}  
\SW_{x,w}(\zeta) - \hat{\SW}_{x,w}(\zeta)\hspace{0.04cm}\right|\le \kappa/5.$$%   is bounded by $n^k$.
   %The bounds on $\{p_\ell\}_{\ell=0}^{n-k}$ follow from 
% \cnote{Added the bound}
% \anote{Sandip: what is the relevant lemma here.} \snote{@Anindya For this section, we don't need a lemma; the bound of $n^k$ follows trivially from the definition. See the paragraph just after \Cref{thm:complex}. In \Cref{sec:strong-version}, \Cref{thm:uniqueness-improvement} requires the stronger bound.} 
%Thus, line~5 of the algorithm returns some solution.  
 
Next we let $q_0, \ldots, q_{n-k}$ be any feasible solution to the linear program
  and assume for a contradiction that 
$|q_0-\SW_{x,w}(0)|\ge1/2$.  
Let $q$ be the polynomial $q(\zeta)=\sum_{\ell=0}^{n-k} q_\ell\hspace{0.03cm}\zeta^\ell$.
  %If the algorithm returns an incorrect answer\cnote{This is not very clear to me.  The algorithm returns the integer closest to $r_0$ but not $r_0$, right?}, that 
%  We consider the polynomial 
% $$r(\zeta) = \sum_{\ell =0}^{n-k} r_\ell\hspace{0.03cm} \zeta^{\ell},\quad\text{where %$r_\ell = p_\ell-q_\ell$ for each $\ell$.}$$ Observe that crucially
%\begin{equation}~\label{eq:r0half}
%$|r_0| \ge 1/2.$
%\end{equation}
Given that $0\le q_\ell\le n^k$ for every $\ell$ (as required by the linear program), 
  \Cref{thm:uniqueness-info-theory} implies (using the choice of $\kappa$ in
line~1 of the algorithm) that
\begin{equation}~\label{eq:r1half}
  \sup_{\zeta \in [\delta,(\delta+1)/2]} \big\lvert \SW_{x,w}(\zeta) - q(\zeta) \big\rvert \geq \kappa.
    \end{equation}
The following claim (with $s=\SW_{x,w} -q$ and $m=n^k$) 
  shows that there exists a $\zeta\in S$
  such that 
$$
\big\lvert \SW_{x,w}(\zeta) - q(\zeta) \big\rvert \geq \kappa/2,
$$
a contradiction to the assumption that $q_0,\ldots,q_{n-k}$ is a feasible solution
  because
$$
\left|\hspace{0.04cm}\sum_{\ell=0}^{n-k} 
q_{\ell} \zeta^{\ell} - \hat{\SW}_{x,w}(\zeta)\hspace{0.04cm}\right|
=\big|q(\zeta)- \hat{\SW}_{x,w}(\zeta)\big|\ge 
\big|q(\zeta)-\SW_{x,w}(\zeta)\big|-\big|\SW_{x,w}(\zeta)-\hat{\SW}_{x,w}(\zeta)\big|
>\kappa/5.$$
%We next have the following lemma. 
 \begin{claim} [Searching over $S$ suffices] \label{lem:granularity}
Let 
$
s(t) = s_0 + s_1 t + \cdots + s_n t^n
$
be a polynomial such that  every coefficient $s_\ell$ has $|s_\ell| \leq m$. 
%, (2) the degree-0 coefficient $s_0$ has $|s_0|\geq 1$.
Suppose $\abs{s(t_0)} \ge \kappa$ for some $t_0 \in [\delta, (\delta+1)/2]$.  Then there exists an integer $k$ such that $t' = \delta + k\Delta \in [\delta, (\delta+1)/2]$ and $\abs{s(t')} \ge \kappa/2$, where $\Delta = \kappa/(2mn^2)$.
\end{claim}
%\begin{comment}
% Fix $0 < \delta_0 < 1,$ and suppose that $|s(t)|>\kappa$ for some $t \in [\delta_0,1]$. Then $|s(t')| > \kappa/2$ for some $t' \in [\delta_0,1]$ which is of the form $t' = $(integer)$\cdot \Delta$ where $\Delta=(\delta_0 \cdot \kappa)/(2mn)$. 
%\end{comment}
\begin{proof}
  Let $k$ be an integer such that $t' := \delta + k\Delta \in [\delta, (\delta+1)/2]$ and $|t'-t_0| \le \Delta$.  Since $|t_0| \le 1$ and $|t'| \le 1$, for each $\ell \in \{1, \ldots n\}$ we have that
  \[
    |t'^\ell - t_0^\ell|
    \le  |t' - t_0| \cdot \sum_{i=0}^{\ell-1} \bigl| t'^i t_0^{\ell-1-i} \bigr|
    \le \Delta \ell
    \le \Delta n .
    %= |(t_0 + \eps)^\ell - t_0^\ell|
    %\le \sum_{i=1}^\ell \red{{\ell \choose i}}\cdot \bigl| t_0^{\ell-i} \eps^i \bigr|
    %\le \sum_{i=1}^\ell {\ell\choose i}\cdot |\eps|^i 
    %= \red{(1+|\eps|)^\ell -1}\le 2|\eps|\ell\le 2n\Delta.%\sum_{i=1}^\ell \Delta^i 
    %\le 2\Delta,
  \]
  %because $|t_0| \le 1$ and $(1+|\eps|)^\ell\le e^{|\eps|\ell}\le 1+2|\eps|\ell$ given
  %that $|\eps|\ell\le \Delta n$ is small.
  Since $|s_\ell| \le m$ and $\Delta = \kappa/(2mn^2)$, we have 
  \[
    \left|s_\ell t'^\ell - s_\ell t_0^{\ell}\right| = \big|s_\ell\big|\cdot \big |t'^\ell - t_0^\ell\big| \le mn \Delta = \kappa/(2n) .
  \]
  Therefore
  \[
    \abs{s(t') - s(t_0)}
    \le \sum_{\ell=1}^n \abs{s_\ell t'^\ell - s_\ell t_0^\ell}
    \le \kappa/2 .
  \]
  It follows from the triangle inequality that $\abs{s(t')} \ge \abs{s(t_0)} - \abs{s(t') - s(t_0)} \ge \kappa/2$.
\end{proof}

We now analyze the complexity of the algorithm. Note that for all $\zeta \in S$, we have $1-\zeta \geq (1-\delta)/2$. % of  
  By \Cref{lem:estimator-polynomial}, %\Cref{corr-poly-eval},
the sample complexity is 
\begin{equation}\label{hehebound}
n^{O(1)} \cdot \left(\frac{2}{1-\delta}\right)^{O(k)} \cdot \left(\frac{5}{\kappa}\right)^2 \cdot \log \left(\frac{|S|}{\tau}\right) %\left(\frac{n}{1-\delta}\right)^{O(k)}\cdot \left(\frac{5}{\kappa}\right)^2 \cdot \log \left(\frac{2|S|}{\tau}\right)
= \left(\frac{n^{1/(1-\delta)}}{1-\delta}\right)^{O(k)} \log \left(\frac{1}{\tau}\right).
%\log \left(\frac{1}{\tau}\right)\cdot \left(\frac{n}{1-\delta}\right)^{O(k/(1-\delta))}.
\end{equation}
%$(n/(1-\delta))^{O(k)} \cdot \gamma^{-2} \cdot \log (1/(\gamma \tau))$ -- here $\gamma$ is the parameter in the description of {{\tt Multiplicity}$'_{\text{large-}\delta}$}. Plugging in the value of $\gamma$\cnote{Where can I find the value of $\gamma$?}, the sample complexity 
% simplifies to  $(n/(1-\delta))^{O(k/(1-\delta))} \log (1/\tau)$.  
The running time of the algorithm is (\ref{hehebound}) multiplied by $|S|$
  plus the time needed to solve the linear program.
The former can still be bounded by the same expression on the RHS of
 (\ref{hehebound}) above.
The latter can be bounded by $\poly(n)$ multiplied by the number of bits 
  needed to describe the linear program, which can also be bounded by the RHS
  of (\ref{hehebound}).
%To bound the running time, observe that for each $\zeta \in S$, computing $\hat{\SW}_{x,w}(\zeta)$
% takes time $(n/(1-\delta))^{O(k)} \cdot \gamma^{-2} \cdot \log (1/(\gamma \tau))$. Since we invoke this algorithm for all $\zeta \in S$, the running time for Step~2 of the algorithm is $(n/(1-\delta))^{O(k/(1-\delta))} \log (1/\tau)$. Finally, invoking the polynomial time algorithm for convex programming~\cite{grotschel2012geometric}, the running time of steps 3(a) and 3(b), remains $\mathsf{poly}(|S|,n) = (n/(1-\delta))^{O(k)} \cdot \gamma^{-2} \cdot \log (1/(\gamma \tau))$. 
This proves the claimed upper bounds on the running time and sample complexity, and concludes the proof of \Cref{thm:Multiplicity-large-delta-weak} assuming \Cref{thm:uniqueness-info-theory}.

\subsection{Proof of \Cref{thm:complex}} \label{sec:proof:thm_complex}

In this subsection we prove \Cref{thm:complex}. For convenience we define $\rho := 1-\delta \in (0,1)$, and we restate the theorem below in terms of $\rho$: 

\medskip

\begin{comment}
\noindent {\bf Restatement of \Cref{thm:complex}}: %(Correctness of the reconstruction for general $\delta$}
\emph{
Let $1 \leq n \leq m$.  Let $p(z) = \sum_{i=0}^n p_\ell z^i$ be a polynomial of degree at most $n$ with real coefficients such that $|p_0| \geq 1/3$, and $|p_\ell|\le m$ for all $i$. Then for any $\delta \in (0,1)$, 
\[
\sup_{\zeta \in [\delta,(\delta+1)/2]} |p(\zeta)| \geq m^{-O\left(1/(1-\delta)\right)}.
\]
%for some positive constant $c(\delta)>0$. 
}
\end{comment}

%It suffices to prove the lemma for sufficiently small $\delta$. 
\noindent {\bf Restatement of \Cref{thm:complex}}: %(Correctness of the reconstruction for general $\delta$}
\emph{
Let $1 \leq n \leq m$.  Let $p(z) = \sum_{i=0}^n p_i z^i$ be a polynomial of degree at most $n$ with real coefficients such that $|p_0| \geq 1/2$, and $|p_i|\le m$ for all $i$. Then for any $\rho \in (0,1)$, 
\[
    \sup_{\zeta \in \left[1-\rho,1- {\rho}/{2}\right]} \big|p(\zeta)\big| \geq m^{-O\left(1/\rho\right)}.
\]
}

The proof uses the Hadamard three-circle theorem, along with other standard results in complex analysis. Consider the mapping $w : \C \rightarrow \C$ given by
\[
w(z) = 1-\frac{3\rho}{4} + \frac{\rho}{8} \left(z + \frac{1}{z}\right).
\]
We observe that the map $w(z)$ is meromorphic with only one pole %\rnote{my complex analysis is not good:  is $z=0$ an ``essential singularity'' or a ``pole''?}
at $z=0$. 
Define radii
\[
r_1=1 ; \ \ r_2=2; \ \ r_3=4.
\]
For $i  = 1,2,3$, let $C_i \subset \C$ be the circle centered at the origin with radius $r_i$. %\snote{Originally, the facts about $w(C_i)$ were here. I shifted the main Hadamard 3-circle theorem application before those because I felt that this will give a better overview of the proof. Please let me know if I should revert. {\bf Rocco: I think this is a good organization.}}
Consider the map $f : \C \rightarrow \C$ given by $f(z) = p(w(z))$. Like $w(\cdot)$, $f$ is meromorphic with only one pole at $z=0$. 
The idea of the proof is to use the Hadamard three-circle theorem \cite{Had-Wiki} on $f$, which tells us that
\begin{equation} \label{eq:hadamard}
2 \log \left(\sup_{z \in C_2} |f(z)|\right) \le  \log \left(\sup_{z \in C_1} |f(z)|\right) + \log 
  \left(\sup_{z \in C_3} |f(z)|\right).
\end{equation}

Now, we will analyze each term in the above inequality. We first record some facts about the behaviour of $w$ over each circle $C_i$ that are immediate from the definition:

\begin{fact} \label{fact_w_locus}
Let $w, C_1,C_2$ and $C_3$ be as defined above.
    \begin{enumerate}[(1)]
    \item When $z$ ranges over $C_1$, $w(z)$ ranges over the real line segment 
    $[1-\rho, 1-\rho/2]$.%, which is our domain of interest.
  
    \item When $z$ ranges over $C_2$, $w(z)$ ranges over the ellipse $E_2$ in the complex plane which is centered at the real value $1-3\rho/4$ and is the locus of all points $z=x+iy$ satisfying
    \[
    \left(\frac{x - (1-3\rho/4)}{5\rho/16}\right)^2 + \left(\frac{y}{3\rho/16}\right)^2=1.
    \]
   
   \item Similarly, when $z \in C_3$, $w(z)$ ranges over the ellipse $E_3$ in the complex plane which is centered at the real value $1-3\rho/4$ and is the locus of all points $z=x+iy$ satisfying
    \[
    \left(\frac{x - (1-3\rho/4)}{17 \rho / 32}\right)^2 + \left(\frac{y}{15\rho/32}\right)^2=1.
    \]
    Moreover, the ellipse $E_3$ is completely contained in the unit disk $B_1(0)$.
\end{enumerate}
\end{fact}

\begin{comment}
\begin{enumerate}
 
    \item When $z$ ranges over $C_1$, $w(z)$ ranges over the line segment 
    $[1-\rho, 1-\rho/2]$, which is our domain of interest.
  
    \item When $z$ ranges over $C_2$, $w(z)$ ranges over the ellipse $E_2$ which is centered at $(1-3\rho/4,0)$ and is the locus of all points $z=x+iy$ satisfying
    \[
    \left(\frac{x - (1-3\rho/4)}{5\rho/16}\right)^2 + \left(\frac{y}{3\rho/16}\right)^2=1.
    \]
   
   \item Similarly, when $z \in C_3$, $w(z)$ ranges over the ellipse $E_3$ which is centered at $(1-3\rho/4,0)$ and is the locus of all points $z=x+iy$ satisfying
    \[
    \left(\frac{x - (1-3\rho/4)}{17 \rho / 32}\right)^2 + \left(\frac{y}{15\rho/32}\right)^2=1.
    \]
We note that for $0< \rho < 1,$ the ellipse $E_3$ is completely contained in the unit circle. 
\end{enumerate}

\end{comment}

\begin{comment}
Now, consider the map $f : \C \rightarrow \C$ given by $f(z) = p(w(z))$. Like $w(\cdot)$, this map $f$ is holomorphic except at $z=0$. 
The idea of the proof is to use the Hadamard three-circle theorem on $f$, which tells us that
\begin{equation} \label{eq:hadamard}
2 \log \sup_{z \in C_2} |f(z)| \le  \log \sup_{z \in C_1} |f(z)| + \log \sup_{z \in C_3} |f(z)|.
\end{equation}
\end{comment}

\Cref{eq:hadamard} will be useful to us because of the following simple claim, which is immediate from \Cref{fact_w_locus}, Item (1):
\begin{claim}~\label{clm:a1}
\[
\sup_{z \in C_1} |f(z)| = \sup_{\zeta \in [1-\rho,1-\rho/2]} |p(\zeta)|.
\]
\end{claim}
\begin{comment}
\begin{proof}
\Cref{fact_w_locus}, item 1 above implies that 
\[
\sup_{z \in C_1} |f(z)| = \sup_{\zeta \in [1-5\rho/8,1-3 \rho/8]} |p(\zeta)|.
\]
The claim is now immediate. 
\end{proof}
\end{comment}

Given \Cref{eq:hadamard} and \Cref{clm:a1}, in order to lower bound $\sup_{\zeta \in [1-\rho,1-\rho/2]}|p(\zeta)|$, it suffices to upper bound $\sup_{z \in C_3} |f(z)|$ and to lower bound $\sup_{z \in C_2} |f(z)|$. We do this in the following claims:

\begin{claim}~\label{clm:a2}
\[
\sup_{z \in C_3} |f(z)| \le m \cdot (n + 1).  
\]
\end{claim} 
\begin{proof}
  By \Cref{fact_w_locus}, Item (3) above, we have $E_3 \subseteq B_1(0)$ and so \ignore{
  %\xnote{Sorry that I am a bit confused about the notation of $B_1(0)$. Is this meant to be the circle or the disc? Earlier  we used it for the disk, i.e., every point with radius at most $1$ from $0$. Maybe here is about the circle?}\cnote{I think it's the disc here.}
  }
  \[
\sup_{z \in C_3} |f(z)| = \sup_{z \in E_3} |p(z)| \leq \sup_{z \in B_1(0)} |p(z)|,
\]
%where the inequality is by the Maximum Modulus Principle\xnote{Do we need to use it here?
%  Can we just say that every $z\in E_3$ has $|z|\le 1$ and then $|p(z)|\le m(n+1)$.
%  Just to make sure that I am not missing something here.}
The bounds on the coefficients of $p$ immediately imply that $\sup_{z \in B_1(0)} |p(z)| \leq m \cdot (n + 1)$.
\end{proof}

\begin{claim}~\label{clm:a3}
\[
\sup_{z \in C_2} |f(z)| \ge m^{-O(1/\rho)}. 
\]
\end{claim}
\begin{proof}
%It is not difficult to see that for $\theta^{\ast} = \Theta(\rho)$, the arc of radius\rnote{Playing with the constants, is $1-\rho$ quite big enough for the origin-centered circle to intersect the ellipse $E_2$? It seems to me we need to take it to be $1-0.6\rho$ or $1-0.51 \rho$ or something like that} \red{$1-\rho$}, centered at the origin spanning angle $[-\theta^{\ast}, \theta^{\ast}]$ lies inside $E_2$ (call this arc $\mathcal{A}$).
Applying Jensen's formula \cite{Jensen-Wiki} to $p$ on the closed origin-centered disk of radius $1 - 3\rho/4$,
%\rnote{(my poor complex analysis cropping up again\dots) Reading the Wikipedia page on Jensen's formula it's not immediately clear to me how this implies the inequality.  Maybe we can cite Fact~5.3 of the DOS trace reconstruction paper, I think that is exactly what is being used below, right?}
%\snote{Yeah I got confused again; it seems like we need to ensure all roots of $p$ lie in the unit circle to apply Jensen; I replaced with the fact from DOS. Please edit this citation}
%a standard fact from the theory of “Mahler measures” (see Fact 5.3 in \cite{DOS17} for a proof),
we get that 
\begin{equation} \label{eq:jensen}
\mathbf{E}_{\bz} [\ln |p(\bz)|] \ge \ln |p(0)| \geq \ln (1/2) = - \ln 2.
\end{equation}
Here $\bz$ is taken to be a uniform random point on the circle $C$ \ignore{
%\xnote{Maybe we use $C$ for circles and $B$ for disks? Rocco: I agree we should use this notation if we need it}
}of radius $1-3\rho/4$ centered at the origin.

Now, consider the arc \[
\calA = \{z \in \C : |z| = 1-3\rho/4 \text{ and } |\arg(z)| \leq 3\rho/16\}.\]
Let $c_{\max, \mathcal{A}} = \max_{z \in \calA} |p(z)|$ and $\theta^{\ast} = 3\rho/16$ (note that $\theta^{\ast}/\pi$  is the fraction of $C$ that lies in ${\cal A}$). 
Now since $|p(z)| \leq m(n+1)$ for all $z \in B_{1-3\rho/4}(0) \setminus \mathcal{A}$ (because of the coefficient bound on $p$), we have by \Cref{eq:jensen} that

\begin{comment}
\[
\left(1 - \frac{\theta^{\ast}}{\pi}\right)\ln \left(m(n+1)\right) + \ln c_{\max, \mathcal{A}} \cdot \frac{\theta^{\ast}}{\pi} \ge 0, \quad \text{and hence} \quad
\ln \left(m(n+1)\right) + \ln c_{\max, \mathcal{A}} \cdot \frac{\theta^{\ast}}{\pi} \ge 0.
\]
\end{comment}

\[
-\ln 2 \leq \left(1 - \frac{\theta^{\ast}}{\pi}\right)\ln \left(m(n+1)\right) + \frac{\theta^{\ast}}{\pi} \cdot \ln c_{\max, \mathcal{A}} \leq \ln \left(m(n+1)\right) + \frac{\theta^{\ast}}{\pi} \cdot \ln c_{\max, \mathcal{A}}.
\]
Thus, 
\[
\ln c_{\max, \mathcal{A}} \ge  -\frac{\pi \cdot \ln \left(2m(n+1)\right)}{\theta^{\ast}}, 
\]
and hence%Thus, 
\[
c_{\max,\mathcal{A}} \ge \left(2m(n+1)\right)^{-\pi/\theta^\ast}. 
\]

Next, we observe that the arc $\calA$ is entirely in the interior of the ellipse $E_2$. (To see this, observe that the center of the arc is the real value $1-3\rho/4$, which coincides with the center of the ellipse, and that every point on the arc is within distance less than $3\rho/16$ from the center of the arc (ellipse).  Since $3\rho/16$ is the length of the semi-minor axis of the ellipse, it follows that every point in the arc is within the ellipse.)
We further recall that $m \geq n$ and that $\theta^\ast = \Theta(\rho)$. Using these facts along with the maximum modulus principle and \Cref{fact_w_locus} Item (2), we conclude that % and recalling that $m \geq n$, we conclude %fact that the max of a holomorphic functions occur at the boundary, we have 
\[
    \sup_{z \in C_2} |f(z)| = \sup_{z \in E_2} |p(z)| \geq \sup_{z \in \mathcal{A}} |p(z)| = c_{\max, \calA} \geq m^{-O(1/\rho)},
\]
and \Cref{clm:a3} is proved.
\begin{comment}
$$
\max_{z \in E_2} |p(z)|  \ge \max_{z \in \mathcal{A}} |p(z)| = c_{\max,\mathcal{A}}.
$$
Observe that 
\[
\max_{z \in E_2} |p(z)| = \max_{z \in C_2} |f(z)|. 
\]
This implies the claim (recalling that $m \geq n$).
\end{comment}
\end{proof}

\begin{proofof}{\Cref{thm:complex}} We combine \Cref{clm:a1,,clm:a2,clm:a3} in \Cref{eq:hadamard} to get that
\[
\log \sup_{\zeta \in [1-\rho,1-\rho/2]} |p(\zeta)| = \log \sup_{z \in C_1} |f(z)| \ge -O(1/\rho) \log m - \log (m(n+1)) \geq -O(1/\rho) \log m.
\]
Exponentiating both sides finishes the proof of \Cref{thm:complex}.
\end{proofof}

%!TEX root = main.tex

\section{Improved algorithms: Proof of \Cref{thm:Multiplicity}} \label{sec:strong-version}

In this section we give improved algorithms strengthening the quantitative bounds given in \Cref{thm:Multiplicity-small-delta-weak} and \Cref{thm:Multiplicity-large-delta-weak} %\rnote{Were we going to try to do a more-efficient-version of \Cref{thm:Multiplicity-small-delta-weak} as well?}
and thereby complete the proof of \Cref{thm:Multiplicity}.%\rnote{\Cref{thm:Multiplicity} will be the one with sample complexity $\exp(O_\delta(k)) n^{O(1)}$.}

First we describe the main ideas underlying the improved algorithms. Both algorithms benefit from the same insights, so we will just describe the improvement of \Cref{thm:Multiplicity-large-delta-weak} in this overview. 
%\rnote{I commented out the ``Assume for the sake of this discussion that $\delta$ is bounded away from $1$, i.e., $\delta = 1-\Omega(1)$; this assumption will not be necessary for any eventual results.'' --- I don't think that that assumption leads to much simplification for the reader}
%Assume for the sake of this discussion that $\delta$ is bounded away from $1$, i.e., $\delta = 1-\Omega(1)$; this assumption will not be necessary for any eventual results. 
Recall the definition of the subword polynomial $\SW_{x,w}$ from \Cref{def:subword-polynomial}:
\[
    \SW_{x,w}(\zeta) := \sum_{\mathclap{\substack{\alpha \in \Z_{\geq 0}^{k-1} \\ |\alpha| \leq n-k}}} \, \#\left(w_0 \ast^{\alpha_1} w_1 \ast^{\alpha_2} w_2 \dots w_{k-2} \ast^{\alpha_{k-1}} w_{k-1}, \hspace{0.06cm}x\right) \cdot \zeta^{|\alpha|}.
\]
\begin{comment}
\begin{equation} \label{eq:subword-poly-degree}
    \SW_{x,w}(\zeta) := \sum_{\ell = 0}^{n-k} \left( \sum_{\alpha \in \Z_{\geq 0}^{k-1}, |\alpha| = \ell} \#\left(w_0 \ast^{\alpha_1} w_1 \ast^{\alpha_2} w_2 \dots w_{k-2} \ast^{\alpha_{k-1}} w_{k-1}, x\right)\right) \, \zeta^{\ell}.
\end{equation}
\end{comment}

Grouping terms of the same degree together, we can write it as $
\SW_{x,w}(\zeta) = \sum_{\ell \geq 0} \gamma_\ell \, \zeta^\ell
$, where
\begin{equation*} \label{eq:subword-poly-degree}
    \gamma_\ell = \sum_{\substack{\alpha \in \Z_{\geq 0}^{k-1}\\ |\alpha| = \ell}} \#\left(w_0 \ast^{\alpha_1} w_1 \ast^{\alpha_2} w_2 \dots w_{k-2} \ast^{\alpha_{k-1}} w_{k-1},\hspace{0.06cm} x\right)
\end{equation*}
is the degree-$\ell$ coefficient, for each $0 \leq \ell \leq n-k$. 
%For each degree $\ell \in {0,1,\cdots,n-k}$, consider the degree-$\ell$ coefficient of $\SW_{x,w}$.
In the proofs of Corollary \ref{corr-poly-eval} in \Cref{sec:small-delta} and \Cref{thm:uniqueness-info-theory} in \Cref{sec:large-delta}, we bounded these coefficients uniformly by $m=n^k$. The first insight is that in fact a sharper bound holds for these coefficients: specifically, we have
\begin{equation} \label{eq:coefficient-bound}
    0 \leq \gamma_\ell \leq m_\ell:=n \binom{\ell + k-2}{k-2}.
\end{equation}
This is simply because there are at most $n$ choices for the position of the first character $w_0$ in $x$, and there are   $\binom{\ell + k-2}{k-2}$ %\cnote{I think this is $\binom{\ell+k-2}{k-2}$?}
ways to choose a tuple of non-negative integers $\alpha_1, \cdots, \alpha_{k-1}$ that sum to $\ell$.  The second insight is that since our approaches only involve evaluating $\SW_{x,w}(\zeta)$ on non-negative real inputs $\zeta$ that are bounded below 1, we can exploit this improved coefficient bound to \emph{truncate} the high-degree portion of the polynomial; working with the resulting (much) lower-degree polynomial leads to an overall gain in efficiency.  

To explain this in more detail, we need the following definition:%\snote{better names? also I would like to keep the definition here so that the next part of the overview makes sense, but if there is a better way to write the overview and push the formal definition to the next subsection, I would be happy with that.{\bf Rocco: I like having the definition here}}
\begin{definition}
    Let $p(\zeta) = \sum_{\ell=0}^n p_\ell \, \zeta^\ell$ be a univariate polynomial of degree at most $n$. For $d \in \{0,1,\cdots,n\}$, we define the \emph{$d$-low-degree part of $p$} (denoted as $p^{\leq d}$) to be
    \[
        p^{\leq d}(\zeta) = \sum_{\ell = 0}^d p_\ell \, \zeta^\ell.
    \]
    Analogously, we define the \emph{$d$-high-degree} part of $p$ to be $p^{>d}(\zeta) := \sum_{\ell >d} p_\ell \, \zeta^\ell = p(\zeta) - p^{\leq d}(\zeta)$.
\end{definition}

Consider any polynomial $q$ with a constant term which is an integer different from $\SW_{x,w}(0)$. In order for $q$ to be a polynomial that could possibly arise from the $k$-subword deck of some string $z \in \zo^n$, it must also have coefficients bounded by the right hand side of \Cref{eq:coefficient-bound}. Using these sharper bounds on  the coefficients, %along with the fact that (by \Cref{thm:uniqueness-info-theory}) $\SW_{x,w}$ needs to be evaluated only at real values $\zeta \leq (\delta+1)/2 = 1 - (1-\delta)/2$,
we show that there exists a threshold degree $d$ that is roughly\footnote{We ignore the
  dependence on $\delta$ for the overview here; see (\ref{choiced2}) and (\ref{choiced}) for exact choices of $d$.} 
  $O( k+\log n)$ such that
  %$d = \Theta_{\delta}(k)$,
%\footnote{In this overview, we assume $k \geq \log n$ to simplify the exposition. No assumptions will be used subsequently.} %, to simplify the exposition. Subsequent sections do not {\bf Rocco:}
%\rnote{I'm not sure how the overview is assuming this? If we are assuming it, should we mention it in a footnote? \textbf{Sandip}: Does this work?} %\snote{Should I remove the $\_\delta$ because $1-\delta$ ia assumed to be constant for this overview?{\bf Rocco: I like the idea of keeping the $\_\delta$, that way we don't need to make the earlier assumption that $1-\delta$ is constant}}
%such that:
\begin{itemize}
    \item The $d$-low-degree part of the polynomials $\SW_{x,w}$ and $q$ %(whose coefficients are bounded by $n \binom{d + k-2}{k-2} \approx 2^{O(k)} n \ll n^k$)
    must differ by at least %\rnote{Was ``$\approx (1-\delta)^{O(k)}$''} 
    %\red{$\approx \left(\frac{2}{1-\delta}\right)^{-O(k)/(1-\delta)}$} %\red{$\approx  \left(n \binom{d+k-2}{k-2}\right)^{-O(1)/(1-\delta)}$}
    $$
    \left(\frac{1}{n} \left(\frac{1-\delta}{2}\right)^k\right)^{O(1/(1-\delta))}
    $$
    (see \Cref{eq:lower-bound-low-degree-part}) at some point in the interval $[\delta,(\delta+1)/2]$. This result is stronger than the analogous $\approx n^{-O(k/(1-\delta))}$ lower bound established in \Cref{thm:uniqueness-info-theory}, which leads to savings on both time and sample complexity. 
    \item The maximum value that the high-degree part of such polynomials attains on the relevant interval is negligible compared to the difference specified above.
\end{itemize}
Combining these two facts enables us to carry out our analysis just on the $d$-low-degree part, which has much smaller coefficients and thereby admits a more efficient algorithm.

In \Cref{sec:improvement_small_delta}, we implement these ideas to strengthen \Cref{thm:Multiplicity-small-delta-weak} when $\delta < 1/2$. In \Cref{sec:improvement_technical}, we do the same to derive a stronger analogue of \Cref{thm:uniqueness-info-theory}, which reduces the sample complexity of computing $\#(w,x)$ for general $\delta < 1$ significantly. Finally in \Cref{sec:improvement_algorithm}, we obtain an LP-based algorithm to compute $\#(w,x)$ which is faster than the corresponding algorithm in \Cref{sec_compute_multiplicity_EFFICIENT}.

\begin{comment}
\blue{we have
\[
    \gamma_\ell \leq n \binom{\ell + k-2}{k-2}.
\]
To see this, note that there are at most $n$ choices for the position of the first character $w_0$ in $x$, and there are at most $\binom{\ell + k-2}{k-2}$ choices for non-negative integers $\alpha_1, \cdots, \alpha_{k-1}$ with sum $\ell$.}
\end{comment}

\subsection{Improvement of \Cref{thm:Multiplicity-small-delta-weak} for deletion rate $\delta < 1/2$} \label{sec:improvement_small_delta}

%\begin{fact}
%When $\ell\ge k$, we have 
%$$
%{\ell+k\choose k}\le \left(\frac{2e\ell}{k}\right)^k,
%$$
%\end{fact}

In this subsection we strengthen \Cref{thm:Multiplicity-small-delta-weak} for deletion rate $\delta < 1/2$ as follows:

\begin{theorem} \label{thm:Multiplicity-small-delta-strong}
Let $0 < \delta < 1/2$. 
%where $c>0$ is an absolute constant.  
There is an algorithm {\tt Multiplicity}$_{\text{small-}\delta}$ which takes as input a string $w \in \zo^k$, %\blue{where $\log n \leq k \leq n$,}
access to independent traces of an unknown source string $x \in \zo^n$, and a parameter $\tau>0$.
{\tt Multiplicity}$_{\text{small-}\delta}$ draws $\poly(n) \cdot (1/2-\delta)^{-O(k)} \cdot \log (1/\tau)$ traces from $\Del_\delta(x)$, runs in time $\poly(n) \cdot (1/2-\delta)^{-O(k)} \cdot \log (1/\tau)$, and has the following property: For any unknown source string $x \in \zo^n$, with probability at least $1-\tau$ the output of {\tt Multiplicity}$_{\text{small-}\delta}$ is the multiplicity of $w$ in $\subword(x,k)$ \emph{(}i.e.~the number of occurrences of $w$ as a subword of $x$\emph{)}.
\end{theorem}

%\begin{proof}
    Recall \Cref{thm:Taylor-subword-poly}, which relates the subword polynomial value at any point $\zeta \in \C$ to traces drawn from the deletion channel using Taylor series:
    \begin{equation*}
        %\label{eqn:subword-poly-eval}
        \SW_{x,w}(\zeta) = \frac{1}{(1-\delta)^k} \; \sum_{\mathclap{\substack{\alpha \in \Z_{\geq 0}^{k-1} \\ |\alpha| \leq n-k}}} \, \E_{\by \sim \Del_\delta(x)}\hspace{-0.04cm} \Big[ \#(w_0 \ast^{\alpha_1} w_1 \ast^{\alpha_2} w_2 \cdots w_{k-2} \ast^{\alpha_{k-1}} w_{k-1},\hspace{0.06cm} \by) \Big]\cdot  \left( \frac{\zeta-\delta}{1-\delta} \right)^{|\alpha|}.
    \end{equation*}
%Consider the polynomial $p = (1-\delta)^k \SW_{x,w}$. This is a polynomial of degree at most $n-k$.
As in \Cref{sec:improvement_small_delta}, our goal is to evaluate $\SW_{x,w}(0) = \#(w,x)$ up to error $1/3$ in magnitude, and return the integer nearest to our estimate. Let $\xi = (\zeta - \delta) / (1-\delta)$, so that $\zeta = \delta + \xi (1-\delta)$. %Collecting all terms of equal degree together, we can express $\SW_{x,w}(\zeta) = (1/(1-\delta))^k \sum_\ell \gamma_\ell \, \xi^\ell$, where
%\[
%    \gamma_\ell = \sum_{\alpha \in \Z_{\geq 0}^{k-1}, |\alpha| = \ell} \E_{\by \sim \Del_\delta(x)} \left[\#\left(w_0 \ast^{\alpha_1} w_1 \ast^{\alpha_2} w_2 \dots w_{k-2} \ast^{\alpha_{k-1}} w_{k-1}, \by\right)\right].
%\]
Consider the polynomial $p$ defined as follows:
\[
    p(\xi) := (1-\delta)^k \cdot \SW_{x,w}\big(\delta + \xi (1-\delta)\big).
\]
%We can express $p(\xi) = %(1-\delta)^k \, \SW_{x,w}(\zeta) =
%\sum_\ell p_\ell \, \xi^\ell$, with coefficients
%\[
%    p_\ell = \sum_{\mathclap{\substack{\alpha \in \Z_{\geq 0}^{k-1} \\ |\alpha| = \ell}}} \ \  \E_{\by \sim \Del_\delta(x)} \Big[\#\left(w_0 \ast^{\alpha_1} w_1 \ast^{\alpha_2} w_2 \dots w_{k-2} \ast^{\alpha_{k-1}} w_{k-1}, \hspace{0.06cm}\by\right)\Big].
%\]
%Using the same argument as that following \Cref{eq:coefficient-bound}, we have $0 \leq p_\ell \leq m_\ell$.
%n \binom{\ell + k-2}{k-2}$.
We have that $\SW_{x,w}(0) = (1-\delta)^{-k}\, p(-\delta/(1-\delta)),$ so estimating $\SW_{x,w}(0)$ up to error $\pm 1/3$ is equivalent to estimating $p(-\delta/(1-\delta))$ up to error $\pm (1-\delta)^{k}/3$. As $0 < \delta < 1/2$, we have $1-\delta > 1/2$, and so it suffices to estimate $p(-\delta/(1-\delta))$ up to error $2^{-k}/3$. Moreover, we have $|-\delta/(1-\delta)| = \delta/(1-\delta) < 1$. We will use these observations to bound the contribution of the high-degree-part of $p$. Let $\theta = 1/2 - \delta$, so that
  $\delta/(1-\delta)\le 2\delta=1-2\theta$.
%\[
%    \frac{\delta}{1-\delta} = \frac{1/2-\theta}{1/2+\theta} = \frac{1-2\theta}{1+2\theta} \leq 1-2\theta.
%\]

\begin{lemma} \label{lem:improvement_small_delta}
    Let $\delta < 1/2$, and let $p$ and $\theta$  %\rnote{Very minor - I removed ``$p, \theta, k$'' since we're not really stipulating anything about $k$, right?}
    be as above. Then by setting 
    \begin{equation}\label{choiced2}
      d := \frac{C}{\theta} \left(k \ln \frac{C}{\theta} + \ln n\right) 
    \end{equation}
    with $C=e^2$, we have  
    \[
        \sup_{|\xi| \leq 1-2\theta} |p^{>d}(\xi)| \leq \frac{0.1}{2^k}.
    \]
\end{lemma}

Before proving \Cref{lem:improvement_small_delta}, we show that it implies \Cref{thm:Multiplicity-small-delta-strong}.

\begin{proof}[Proof of \Cref{thm:Multiplicity-small-delta-strong} assuming \Cref{lem:improvement_small_delta}] Consider $p^{\leq d}$, the $d$-low-degree-part of $p$, where $d$ is as given by \Cref{lem:improvement_small_delta}. For all $\xi$ with $|\xi| \leq 1-2\theta$,
\[
    |p(\xi) - p^{\leq d}(\xi)| = |p^{>d}(\xi)| \leq \frac{0.1}{2^k}.
\]
So, by the triangle inequality, in order to estimate $p(-\delta/(1-\delta))$ up to error $\pm 2^{-k}/3$, it suffices to estimate $p^{\leq d}(-\delta/(1-\delta))$ up to error $\pm 2^{-k}/5$.

Let $S_d$ be the set $\{\alpha \in \Z_{\geq 0}^{k-1}: |\alpha| \leq d\}$.
As in \Cref{sec_estimate_poly_values}, let
\[
E_\alpha:=\Ex_{\by \sim \Del_\delta(x)}\hspace{-0.04cm} \Big[ \#(w_0 \ast^{\alpha_1} w_1 \ast^{\alpha_2} w_2 \cdots w_{k-2} \ast^{\alpha_{k-1}} w_{k-1},\hspace{0.06cm} \by)\Big]
\]
for each $\alpha \in S_d$. (Note that by definition, $p^{\leq d}$ only includes terms $E_\alpha$ for $|\alpha| \leq d$.) Then
\[
  p^{\le d}(\xi)
  = \sum_{\alpha \in S_d} E_\alpha\cdot \xi^{|\alpha|}.
\]
Each $E_\alpha$ is between $0$ and $n$ and 
  using the same argument as that following \Cref{eq:coefficient-bound}, we have  
%n \binom{\ell + k-2}{k-2}$.
%For each such $\alpha$ and trace $\by \sim \Del_\delta(x)$, we have $\#(w_0 \ast^{\alpha_1} w_1 \ast^{\alpha_2} w_2 \cdots w_{k-2} \ast^{\alpha_{k-1}} w_{k-1},\hspace{0.06cm} \by) \in [0,n]$ and %. The maximum possible cumulative contribution of all such terms (i.e., $\sup_{|\zeta|\leq 1} |p^{\leq d}(\zeta)|$) %$\alpha \in \Z_{\geq 0}^{k-1}, |\alpha| \leq d$
%is at most
\[
  |S_d|= M:= \sum_{\ell=0}^d \binom{\ell + k - 2}{k-2} = \binom{d+k-1}{k-1} \le \binom{d+k}{k}
%|S_d|= M:= \sum_{\ell=0}^d \binom{\ell + k - 2}{k-2} \le (d+1)\cdot \binom{d+k-2}{k-2}
%    \le n\cdot  \binom{d+k}{k}
\]
and we use the following claim to bound the right hand side:

\begin{claim}\label{hehe102}
  Let $d = \frac{C}{\theta} (k\ln\frac{C}{\theta} + \ln n)$ for some $\theta \in (0,1]$ and $C \ge e^2$.  Then we have
  \[ 
    \binom{d+k}{k}\le n\cdot \left(\frac{C}{\theta}\right)^{3k}.
  \]
\end{claim}
\begin{proof}
Using $d\ge k$ and the approximation $k!\ge \sqrt{2\pi k} (k/e)^k\ge (k/e)^k$, we have 
\begin{align}
\binom{d+k}{k}\le \frac{(2d)^k}{(k/e)^k}=
\exp\left(k\ln \frac{2ed}{k} \right)
&\le \exp\left(k\left(2+\ln \frac{C}{\theta}+\ln \left(\ln \frac{C}{\theta}+\frac{\ln n}{k}\right)\right)\right) \nonumber\\ \label{hehe1001}
&\le \exp\left(k\left(2+\ln \frac{C}{\theta}+ \ln \frac{C}{\theta}+\frac{\ln n}{k}\right) \right)\\
&\le n\cdot \left(\frac{C}{\theta}\right)^{3k}, \label{hehe1011}
\end{align}
where (\ref{hehe1001}) used $\ln a\le a$, (\ref{hehe1011}) used $2 < \ln (C/\theta)$ since $C \ge e^2$.
\end{proof}  

Plugging  in \Cref{hehe102}, we have $M\le n /\theta^{O(k)}$ using $\theta<1/2$.
%Also note that 
%$$
%|S_d|=\sum_{\ell=0}^d \binom{d+k-2}{k-2}=M/n.
%$$
%    \\
%     &\leq \exp\left(2 \log n + k \log \frac{2ed}{k} \right)\\
%    &\leq \exp\left(2 \log n + O(k)+k \log \frac{C}{\theta} + k \log \left(\log \frac{1}{\theta} + \frac{\log n}{k}\right)\right)\\
%     &\leq \exp\left(2 \log n + O(k)+ k \log \frac{C}{\theta} + k\log \frac{1}{\theta} + {\log n}\right)\\
%     &\leq n^{O(1)} \left(\frac{C}{\theta}\right)^{O(k)} = \frac{n^{O(1)}}{\theta^{O(k)}}.
%end{align*}
%In the final step above, we have used the fact that $\theta < 1/2$. % and $k \geq \log n$.
The algorithm just draws $s$ (to~be specified)
traces $\by \sim \Del_\delta(x)$, computes an empirical estimate $\tilde{E_\alpha}$ of $E_\alpha$ for each $\alpha\in S_d$ so that 
\[
    \left|\tilde{E_\alpha} - E_\alpha \right| \leq \frac{0.2}{2^kM}.
     %\frac{\theta^{O(k)}}{5 \, n^{O(1)}} 
\]
with probability at least $1-\tau$.
This can be achieved by setting the number of traces to be
\[
  s := O\left( \bigl(M^2 2^k \bigr)^2\right) \cdot \log\left( \frac{M}{\tau}\right)= \left(\frac{n}{\theta^k}\right)^{O(1)} \cdot \log \frac{1}{\tau} 
\]
and a simple application of a Chernoff bound and a union bound. 
When this happens, 
%and combines them to estimate $p^{\leq d}(-\delta / (1-\delta))$. By a simple application of a Chernoff bound and a union bound, we have that with probability at least $1-\tau$, all estimates $\tilde{E_\alpha}$ satisfy
%By choosing a sufficiently large constant in the exponent in the expression of $s$ and using the 
it follows from the fact that $|-\delta / (1-\delta)| < 1$
  that $$\sum_{\alpha\in S_d} \tilde{E_\alpha}\cdot \left(\frac{-\delta}{1-\delta}\right)^{|\alpha|}$$
  is an estimate that deviates by at most $2^{-k}/5$.
%we get that the estimate of $p^{\leq d}(-\delta / (1-\delta))$ obtained as described above deviates from the true value by at most $2^{-k}/5$ with probability at least $1-\tau$. 
Combined with the observations at the beginning of the proof, this implies that we can estimate $\SW_{x,w}(0) = \#(w,x)$ up to error $\pm 1/3$, and hence our output (the nearest integer to our estimate of $\SW_{x,w}(0)$) is $\#(w,x)$ with probability at least $1-\tau$.

The runtime is governed by the time required to compute estimates $\tilde{E_\alpha}$. We can bound it by
\[
  s \cdot n^{O(1)} \cdot |S_d| \leq \left(\frac{n}{\theta^k}\right)^{O(1)} \cdot \log \frac{1}{\tau} = n^{O(1)} \cdot \left(\frac{1}{1/2-\delta}\right)^{O(k)} \cdot \log \frac{1}{\tau}.
\]
This finishes the proof of the theorem.
\end{proof}

\begin{proof}[Proof of \Cref{lem:improvement_small_delta}]
We are interested in $|p^{>d}(\xi)|$ over $|\xi|\le 1-2\theta$, which is trivially bounded by
\[
  |p^{>d}(\xi)|
  \le \sum_{\ell={d+1}}^{n-k} n \binom{\ell+k-2}{k-2} \cdot (1-2\theta)^\ell 
  \le \sum_{\ell=d}^{n-k} n \binom{\ell+k}{k} \cdot (1-2\theta)^\ell.
\]
%Throughout this proof, we will write $k'$ to denote $k-2$. 
First, we show that terms in the sum on the right hand side above decreases with $\ell$ so it 
  suffices to bound the term with $\ell=d$ multiplied by $n$. 
%for the polynomial $p^{>d}$ and our choice of $d$, the magnitude of the maximum possible contribution from the degree-$\ell$ term decreases with $\ell$ for all $\ell > d$, whenever $|\zeta| \leq 1-2\theta$, and so it suffices to bound the degree-($d+1$) term. %This would imply that it suffices to bound the degree-($\ell+1$) term.
    To see this, observe that
    \[
        \left|\frac{ \binom{\ell+k}{k}}{\binom{\ell+k-1}{k}} \cdot (1-2\theta) \right| =
        \frac{\ell+k}{\ell}\cdot (1-2\theta) \le
        %\left(1 + \frac{k }{\ell }\right)\left(1 - 2\theta\right)
        1 + \frac{k }{\ell } - 2\theta < 1,
    \]
    whenever $\ell > k /2 \theta$, which holds for all $\ell > d$ given our choice of $d$. So,
    \[
        \sup_{|\xi| \leq 1-2\theta} \left|p^{>d}(\xi)\right|
        %\leq \sum_{\ell = d+1}^{n-k} n \binom{\ell+k'}{k'} \left(1 - 2\theta\right)^\ell 
        \leq n^2 \binom{d+k}{k} \left(1 - 2\theta\right)^d
        \le n^2 \binom{d+k}{k} e^{-2\theta d}.
        %\leq 2n^2 \binom{d+k'}{k'} \left(\frac{1-\delta}{2}\right)^{d+1}.
    \]
    We have $e^{-2\theta d} = n^{-\frac{2C}{\theta}} \cdot (C/\theta)^{-\frac{2Ck}{\theta}}$, and so plugging in \Cref{hehe102} we have
    \[
      n^2 \cdot \bigl(n \cdot \bigl(C/\theta\bigr)^{3k} \bigr) \cdot e^{-2\theta d}
      \le n^{3 -\frac{2C}{\theta}} \cdot \bigl( C/\theta \bigr)^{(3-\frac{2C}{\theta})k}
      \le \frac{1}{n 2^k}
    \]
    because $3-2C/\theta \le -1$ when $C = e^2$.  This concludes the proof of the lemma.

%Plugging in Claim \ref{hehe102}, we can bound it from above by
%$$
%\exp\left(3\log n+2k \log \frac{C}{\theta}-2\theta d \right)
%=\exp\left(3\log n+2k\log \frac{C}{\theta}-2C\left(k\log \frac{1}{\theta}+\log n\right)\right)
%\le \frac{1}{n2^k}
%$$
%using $C=e^3\approx 20$. This concludes the proof of the lemma.
\end{proof}
%\begin{comment}
%\Hence it is enough to show that
%\begin{equation} \label{eq:getme}
%    n^2 \binom{d + k-1}{k-2} \left(1 - 2\theta\right)^{d+1} \leq \frac{0.1}{2^k}.
%\end{equation}
%Let us define
%\[
%    \psi = \log\left(10 \, n^2 \, 2^k \binom{d + k'+1}{k'} \left(1 - 2\theta\right)^{d+1}\right)
%\]
%Below we prove that $\psi \leq 0$, which gives \Cref{eq:getme} as desired. We will use the facts that %$k' \geq \log n-2$,
%$\log \binom{n}{k} \leq k \log (en/k)$, and $\log (1-x) \leq -x$ for $x < 1$, along with our choice of $d$, for simplification.

%\begin{align*}
%    \psi &\leq k' \log \frac{e(d+k'+1)}{k'} + (d+1) \log (1-2\theta) + 2 \log n + k + \log 10\\
%    &\leq c \, \left(k' \log \frac{d}{k'} + \log n\right) - d \cdot 2\theta \tag{for some constant $c>0$}\\ 
%    &\leq c \, k' \left(\log \frac{C}{\theta} + \log \left(\log \frac{1}{\theta} + \frac{\log n}{k'}\right) + %\frac{\log n}{k'} \right) - 2 \, C\, \left(k' \log \frac{1}{\theta} + \log n\right)\\
%    &\leq c \, k' \left(\log \frac{C}{\theta} + \log \frac{1}{\theta} + \frac{2 \, \log n}{k'}\right) - 2 \, %C\, \left(k' \log \frac{1}{\theta} + \log n\right)\\
%    &\leq 2 \, c \, k' \left(\log \frac{C}{\theta} - \frac{C}{c} \log \frac{1}{\theta}\right) + 2 (c - C) \log n.
%\end{align*}
%By choosing $C$ to be a sufficiently large constant depending on $c$ (but independent of $\theta$), we get that $\psi \leq 0$, which concludes the proof.
%\end{proofof}\end{comment}

\subsection{Improvement of \Cref{thm:uniqueness-info-theory} for deletion rate $\delta < 1$} \label{sec:improvement_technical}

Our main technical result is the following, which is a strengthening of \Cref{thm:uniqueness-info-theory}:
\begin{theorem} \label{thm:uniqueness-improvement}
    Fix $x \in \zo^n$ and $w \in \zo^k$ with $k\le n$. Let $q(z) = \sum_{\ell=0}^{n-k} q_\ell \, z^\ell$ be any polynomial such that %\blue{
    $|\SW_{x,w}(0) - q(0)| \geq 1/2$ 
    %} \rnote{This was ``$q(0)$ is an integer, $q(0) \neq \#(w,x)$'' but the statement of \Cref{thm:uniqueness-info-theory} changed}
    and $0 \leq q_\ell \leq m_\ell$ %n\binom{\ell + k-2}{k-2}$ 
    for all $\ell \in \{0,1,\cdots,n-k\}$. 
    Then 
    \begin{equation} \label{eq:hoho}
        \sup_{\zeta \in [\delta,(\delta+1)/2]} \big\lvert \SW_{x,w}(\zeta) - q(\zeta) \big\rvert \geq
        \left(\frac{1}{n} \left(\frac{1-\delta}{2}\right)^k\right)^{O(1/(1-\delta))},
        \quad\text{for any $\delta \in (0,1)$.} %\left(\frac{(1-\delta)^k}{n}\right)^{O(1/(1-\delta))}.
    \end{equation}
\end{theorem}

%We follow steps similar to the previous section.
Let $p(z)= \SW_{x,w}(z) - q(z)=\sum_{\ell=0}^{n-k}p_\ell\, z^\ell$.
Let $c > 0$ be the constant hidden in the exponent of the RHS of \Cref{final00} in \Cref{thm:complex}.
Let $\theta = (1-\delta)^2/2$.
We will choose the threshold on the degree to be
\begin{equation}\label{choiced}
    d := \frac{C}{\theta} \left(k \ln \frac{C}{\theta} + \ln n\right) 
\end{equation}
where $C=e^2\max(1,c)$.
%and $c>0$ is the constant hidden in the exponent of the RHS of (\ref{final00}) in \Cref{thm:complex}.
%$d = \Theta_\delta(k)$
%on the degree, %which we will specify later,
For this $d$, consider the $d$-low-degree part $p^{\leq d}$. This is a polynomial of degree at most $d$ with $|p^{\leq d}(0)| \geq 1 /2$ and the degree-$\ell$ coefficient is bounded by
\[
  |p^{\leq d}_\ell| \le n\binom{\ell+k-2}{k-2}\le n \binom{d + k-2}{k-2}
  \le n\binom{d+k}{k}
\]
for all $\ell \leq d$. We invoke \Cref{thm:complex} on $p^{\leq d}$ to conclude that
\begin{equation} \label{eq:lower-bound-low-degree-part}
    \sup_{\zeta \in [\delta, (\delta+1)/2]} \left|p^{\leq d}(\zeta)\right| \geq \left(n \binom{d+k }{k }\right)^{-c/(1-\delta)}.
\end{equation}
%for some absolute constant $c > 0$. 
The following lemma upper bounds the contribution of the high-degree part $p^{>d}$ of $p$:
\begin{lemma} \label{lem:upper-bound-high-degree-part}
    Let $p$ and $d$ be as above. Then
\begin{equation}\label{hehe1033}
        %\sup_{\red{\zeta \in B_{(\delta+1)/2}(0)}}
        \sup_{\zeta \in [\delta, (\delta+1)/2]} \left|p^{> d}(\zeta)\right| \leq %\frac{1}{n} 
        \frac{1}{n}\cdot  \left(n \binom{d+k }{k }\right)^{-c/(1-\delta)}.
\end{equation}
    %\rnote{It's true, but is there a reason to write ``$\red{\zeta \in B_{(\delta+1)/2}(0)}$'' above instead of just ``$\zeta \in [0,(\delta+1)/2]$'' or ``$\zeta \in [\delta, (\delta+1)/2]$'' --- we are not using anything about complex numbers per se, right?}
%where $B_{(\delta+1)/2}(0)$ is the complex disc of radius $(\delta+1)/2$ centered at the origin.
\end{lemma}

Before proving this lemma, we show that it implies \Cref{thm:uniqueness-improvement}.

\begin{proof}[Proof of \Cref{thm:uniqueness-improvement} using \Cref{lem:upper-bound-high-degree-part}]  
Since $p = p^{\leq d} + p^{>d}$, %Moreover, $[\delta, (\delta+1)/2] \subset B_{(\delta+1)/2}(0)$.
we use \Cref{lem:upper-bound-high-degree-part} and (\ref{eq:lower-bound-low-degree-part}) to get 
\[
    \sup_{\zeta \in [\delta, (\delta+1)/2]} \left|p(\zeta)\right| \geq  0.9\cdot \left(n \binom{d+k }{k }\right)^{-c/(1-\delta)}.
\]
%We prove the following claim similar to Claim \ref{hehe102} with the new choice of $d$:
%\begin{claim}\label{hehe103}
%We have 
%$$
%\binom{d+k}{k}\le n\cdot \left(\frac{C}{1-\delta}\right)^{3k}.
%$$
%\end{claim}
%\begin{proof}
%Using $d\ge k$ and the approximation $k!\ge \sqrt{2\pi k} (k/e)^k\ge (k/e)^k$, we have 
%\begin{align}
%\binom{d+k}{k}\le \exp\left(k\log \frac{2ed}{k} \right)
%&\le \exp\left(k\left(2+\log \frac{C}{(1-\delta)^2}+\log \left(\log \frac{C}{ 1-\delta }+\frac{\log n}{k}\right)\right)\right) \nonumber\\ \label{hehe100}
%&\le \exp\left(k\left(2+\log \frac{C}{(1-\delta)^2}+ \log \frac{C}{1-\delta}+\frac{\log n}{k}\right) \right)\\
%&\le n\cdot \left(\frac{C}{1-\delta}\right)^{3k}, \label{hehe101}
%%=n\cdot \left(\frac{2}{1-\delta}\right)^{O(k)} 
%\end{align}
%where (\ref{hehe100}) used $\log a\le a$ and (\ref{hehe101}) used $C\ge e^3$.
%\end{proof}  
Plugging in Claim \ref{hehe102} with our choice of $d$, we have 
%Now, with our choice of $d$, we can repeat the calculations in the proof of \Cref{thm:Multiplicity-small-delta-strong} to conclude % and the fact that $k \geq \lg n$,
%\begin{align*}
%    n \binom{d+k-2}{k-2} &\leq \exp\left(\log n + (k-2) \log \frac{e(d+k-2)}{k-2}\right)\\
%    &\leq \exp\left(\log n + O\left(k \log \frac{C}{(1-\delta)^2} + k \log \left(\log \frac{1}{1-\delta} + \frac{\log n}{k}\right)\right)\right)\\
%    &\leq \exp\left(O\left(\log n + k \log \frac{2}{1-\delta}\right)\right) = \left(n \left(\frac{2}{1-\delta}\right)^k\right)^{O(1)}.
%\end{align*}
%So
\[
  \sup_{\zeta \in [\delta, (\delta+1)/2]} \left|p(\zeta)\right|
  \ge 0.9 \left(n \binom{d+k }{k }\right)^{-c/(1-\delta)} \geq  \left(\frac{1}{n} \left(\frac{1-\delta}{2}\right)^k\right)^{O(1/(1-\delta))},
\]
which concludes the proof of \Cref{thm:uniqueness-improvement} using \Cref{lem:upper-bound-high-degree-part}.
\end{proof}

%Now we prove \Cref{lem:upper-bound-high-degree-part}.

\begin{proof}[Proof of \Cref{lem:upper-bound-high-degree-part}]
%\begin{comment}
    This proof is similar to that of \Cref{lem:improvement_small_delta}. %Let $k' = k-2$. 
    First we show that the maximum possible contribution to $p^{>d}(\zeta)$,
    when $\zeta\in [\delta,(\delta+1)/2]$, arises from the degree-$d$ term in $p$: %This would imply that it suffices to bound the degree-($\ell+1$) term.
    %To see this, observe that
    \[
      \left|\frac{ \binom{\ell+k}{k} }{ \binom{\ell+k-1}{k}}\cdot \zeta\right| = 
      \frac{\ell+k}{\ell}\cdot |\zeta| \leq
      \left(1 + \frac{k }{\ell }\right)\left(1 - \frac{1-\delta}{2}\right) \leq 1 + \frac{k }{\ell } - \frac{1-\delta}{2} < 1
    \]
    whenever $\ell > 2k /(1-\delta)$, which holds for all $\ell > d$. So,
    \[
      \sup_{|\zeta| \leq (\delta+1)/2} \left|p^{>d}(\zeta)\right| \leq 
      %\sum_{\ell = d+1}^{n-k} n \binom{\ell+k'}{k'} \left(1 - \frac{1-\delta}{2}\right)^\ell \leq 
      n^2 \binom{d+k}{k} \left(1 - \frac{1-\delta}{2}\right)^d
      \le n^2\binom{d+k}{k} \cdot \exp\left(-\frac{(1-\delta)d}{2}\right).
      %\leq 2n^2 \binom{d+k'}{k'} \left(\frac{1-\delta}{2}\right)^{d+1}.
    \]
It suffices to show that 
\[
  n^2\binom{d+k}{k} \cdot \exp\left(-\frac{(1-\delta)d}{2}\right)\le \frac{1}{n} 
  \left(n \binom{d+k }{k }\right)^{-\frac{c}{1-\delta}} 
\]
or equivalently,  
\begin{align} \label{hehe104}
  n^{3+\frac{2c}{1-\delta}}\cdot \binom{d+k}{k}^{1+\frac{c}{1-\delta}}\cdot \exp\left(-\frac{(1-\delta)d}{2}\right) \le 1.
\end{align}
By our choice of $d$ we have
\[
  \exp\left(-\frac{(1-\delta)d}{2}\right)
  \le n^{-\frac{C}{1-\delta}} \cdot \bigl(C/\theta\bigr)^{-\frac{kC}{1-\delta}} .
\]
Using \Cref{hehe102} again, the left hand side of \Cref{hehe104} is at most
\[
  n^{3 + \frac{2c}{1-\delta} - \frac{C}{1-\delta}} \cdot \bigl(C/\theta\bigr)^{k\left(3 + \frac{3c}{1-\delta} - \frac{C}{1-\delta}\right)} \le 1
\]
because $3 + \frac{3c}{1-\delta} - \frac{C}{1-\delta} \le 0$ when $C = e^2\max(1,c)$.
%$$
%\exp\left(\left(4+\frac{2c}{1-\delta}\right)\log n+ \left(1+\frac{c}{1-\delta}\right)3k\log \frac{C}{1-\delta}-\frac{(1-\delta)d}{2}\right).
%$$
%The exponent is negative given that 
%\[
%    d = \frac{C}{(1-\delta)^2} \left(k \log \frac{C}{1-\delta} + \log n\right) 
%\]
%and $C=e^3\max(1,c)$.
This concludes the proof of the lemma. 
\end{proof}

% !TEX root =  main.tex

\subsection{The algorithm of \Cref{thm:Multiplicity}} %\Cref{def:subword-polynomial}}
\label{sec:improvement_algorithm}

%\red{
Armed with \Cref{thm:uniqueness-improvement} in place of \Cref{thm:uniqueness-info-theory}, the algorithm  {\tt Multiplicity}$_{\text{large-}\delta}$ giving \Cref{thm:Multiplicity} and its analysis are very similar to the algorithm {{\tt Multiplicity}$'_{\text{large-}\delta}$} and its analysis given earlier in \Cref{sec_compute_multiplicity_EFFICIENT}; we only indicate the differences here. 

The algorithm changes in the following ways:

\begin{itemize}

\item In Line~1 of the algorithm, we now set $\kappa$ to be the RHS of \Cref{eq:hoho}:
\[
\kappa :=\left(\frac{1}{n} \left(\frac{1-\delta}{2}\right)^k\right)^{O(1/(1-\delta))}.
\]
%\blue{
With this choice of $\kappa$, it follows from the proof
  of \Cref{thm:uniqueness-improvement} that  the RHS of \Cref{hehe1033} in \Cref{lem:upper-bound-high-degree-part}
  can be bounded from above by $0.01\kappa$.%\xnote{This is to help address footnote 23 of Sandip.}}
%$\left(\frac{(1-\delta)^k}{n}\right)^{O(1/(1-\delta))}$,

\item Later in Line~1, we now set%\xnote{I tried to use Claim 6.2 as a black box so changed the $5$ back to $2$ below.}
$$\Delta := 
\frac{\kappa}{2d^2\, m_d}=
\frac{\kappa}{2d^2 \cdot  n \binom{d + k-2}{k-2} },$$ where $d$ is as given in \Cref{choiced} %\blue{
(the idea is that now we are using the sharper coefficient bound %$n \binom{\ell + k-2}{k-2}$ 
  $m_\ell\le m_d$
given by \Cref{eq:coefficient-bound} rather than the cruder $n^k$ bound used earlier).%}\snote{The blue part in brackets is confusing (because $\ell$ is unspecified) and also probably unnecessary since we are using the sharper bounds explicitly in the next line. I am in favour of removing this line.} 
\begin{comment}
Using the bounds established in \Cref{sec:improvement_technical}, we have
\[
    \Delta \geq \left(\frac{1}{n} \left(\frac{1-\delta}{2}\right)^k\right)^{O(1/(1-\delta))}.
\]
\end{comment} 

\item The coefficient bound on $q_0,\dots,q_{n-k}$ in Line~3(a) for the linear program is now 
$q_\ell \in [0, m_\ell]$ %\blue{m_\ell}]$ 
%\binom{\ell + k-2}{k-2}\big]$ 
for all $\ell \in \{0,1,\cdots,n-k\}$ rather than $q_0,\dots,q_{n-k} \in [0,n^k]$ as earlier. 

\end{itemize}

\begin{comment}
\blue{Assumption: There exists an algorithm (\Cref{lem:estimator-polynomial}) to compute $\SW_{x,w}(\zeta)$ for all $\zeta \in [\delta, 1]$ up to error $\pm \eps$ with probability at least $1-\tau$ using 
\[
    \frac{n^2}{(1-\zeta)^{2k} \, \eps^2} \log \frac{1}{\tau}
\]
traces.}\snote{Leaving it as an explicit assumption till the relevant lemma in \Cref{sec:large-delta} is completed.}
\end{comment}

With these changes to the algorithm, most of the analysis goes through unchanged. As before, we observe that with probability at least $1-\tau$, we have
\[
\text{for 
  every $\zeta\in S$,} \quad \left|\hat{\SW}_{x,w}(\zeta) -{\SW}_{x,w}(\zeta)\right| \le \kappa/5.
\]
We assume this happens henceforth. The solution which sets $q_\ell=(\SW_{x,w})_\ell$, the degree-$\ell$ coefficient of $\SW_{x,w}$, for all $\ell$, is clearly feasible.
%, and in this case, the algorithm returns $\#(w,x)$.

Now we show that every feasible solution $q_0,\cdots,q_{n-k}$ to the linear program must satisfy $|q_0 - \SW_{x,w}(0)| < 1/2$; this is the only part of the analysis that is somewhat different. Suppose for  a contradiction that $q_0,\cdots,q_{n-k}$ is a feasible solution with $|q_0 - \SW_{x,w}(0)| \geq 1/2$. Let $q(\zeta) = \sum_\ell q_\ell \, \zeta^\ell$  and define the polynomial $p = \SW_{x,w}-q$,
%\[
%p(\zeta) = \SW_{x,w}(\zeta) - \sum_{\ell = 0}^{n-k} %q_\ell \, \zeta^\ell,
%\]
with coefficients $p_\ell$. %We have $|p_\ell| \leq m_\ell$ for all $\ell$, and $|p_0|\geq 1/2$. 
We invoke \Cref{thm:uniqueness-improvement} to get that $|p(\zeta^*)| \geq \kappa$ for some $\zeta^\ast \in [\delta, (\delta+1)/2]$. 
By \Cref{lem:upper-bound-high-degree-part} (and the remark below the choice
  of $\kappa$),%\snote{This is fine, right? Xi: I think we don't need to change $d$ now.} %the arguments used to prove \Cref{eq:lower-bound-low-degree-part} and \Cref{lem:upper-bound-high-degree-part} and changing $d$ by a constant factor, we can conclude that
\begin{equation} \label{eq:improved_algo_high_degree}
    \left|p(\zeta) - p^{\leq d}(\zeta)\right| = \left|p^{>d}(\zeta)\right| \leq {0.01\kappa} \end{equation}
for all $\zeta \in [\delta, (\delta+1)/2$]. 
As a result, we have $|p^{\le d}(\zeta^*)|\ge 0.99\kappa$.
Applying~\Cref{lem:granularity} with~$s=p^{\leq d}$, $n=d$, $t_0=\zeta^*$, $m = m_d$ and our choice of $\Delta$, there exists a $\zeta'\in S$ such that 
  $|p^{\le d}(\zeta')|\ge 0.495\kappa$ and thus,
  $|p(\zeta')|\ge |p^{\le d}(\zeta')| - |p^{>d}(\zeta')|\ge0.485\kappa$.
Hence, recalling that $p = \SW_{x,w}-q$, we have
\[
    \left|\hat{\SW}_{x,w}(\zeta') - q(\zeta')\right| \geq \left|p(\zeta')\right| - \left|\hat{\SW}_{x,w}(\zeta') -{\SW}_{x,w}(\zeta')\right| \geq 0.285\kappa>\kappa/5.
\]
As $\zeta' \in S$, the solution $q$ violates a constraint of the LP. This concludes the proof of correctness.

%Now let $\zeta'$ be the closest point in $S$ to $\zeta^\ast$, so that $|\zeta'-\zeta^\ast| \leq \Delta$. So,
%\[
%    \left|p(\zeta') - p(\zeta^*)\right| \leq \left|p(\zeta') - p^{\leq d}(\zeta')\right| + \left|p^{\leq d}(\zeta') - p^{\leq d}(\zeta^*)\right| + \left|p^{\leq d}(\zeta^\ast) - p(\zeta^\ast)\right|.
%\]
%The first and third terms above are at most $\kappa/10$ each by \Cref{eq:improved_algo_high_degree}. To bound the second term, we repeat the arguments in \Cref{lem:granularity} with $s=p^{\leq d}$, $n=d$, $m = n\binom{d+k-2}{k-2}$ and our choice of $\Delta$ to get that $|p^{\leq d}(\zeta') - p^{\leq d}(\zeta^\ast)| \leq \kappa/5$. We conclude that $|p(\zeta') - p(\zeta^*)| \leq 2\kappa/5$, and so
%\[
%    \left|p(\zeta')\right| \geq \left|p(\zeta^\ast)\right| - \left|p(\zeta') - p(\zeta^*)\right| \geq 3 \kappa/5.
%\]
%Hence, recalling that $p = \SW_{x,w}-q$, we have
%\[
%    \left|\hat{\SW}_{x,w}(\zeta') - q(\zeta')\right| \geq \left|p(\zeta')\right| - \left|\hat{\SW}_{x,w}(\zeta') -{\SW}_{x,w}(\zeta')\right| \geq 2\kappa/5.
%\]
%As $\zeta' \in S$, the solution $q$ violates a constraint of the LP. This concludes the proof of correctness.

Now we analyze the sample complexity of the algorithm. We have
\[ 
    |S| = O(1/\Delta) = \left(n \left(\frac{2}{1-\delta} \right)^k\right)^{O(1/(1-\delta))},
\] using the bounds established in \Cref{sec:improvement_technical}. Moreover, all points $\zeta \in S$ satisfy $1-\zeta \geq (1-\delta)/2$. So, by \Cref{lem:estimator-polynomial}, the sample complexity is at most
\begin{equation}\label{hehe1035}
    s = \frac{n^{O(1)}}{\kappa^2} \left(\frac{2}{1-\delta}\right)^{O(k)} \log %\blue{
    \left(\frac{|S|}{\tau}\right)%}
    = \left(n \left(\frac{2}{1-\delta} \right)^k\right)^{O(1/(1-\delta))} \log \frac{1}{\tau}.
\end{equation}
The running time is dominated by the time required to compute $\hat{SW}_{x,w}(\zeta)$ for each $\zeta \in S$.
The running time for each $\zeta$ can be bounded by (\ref{hehe1035}) and the same expression 
   can be used to bound the overall running time given the bound on $|S|$ above.
 %\snote{I think this is true; the LP should take $\poly(n) \cdot S$ time. Please edit this.}, 
%  and can be bounded by
%\[
%    |S| \cdot s = \left(n \left(\frac{2}{1-\delta} \right)^k\right)^{O(1/(1-\delta))} \log \frac{1}{\tau}.
%\]
%\red{With these changes to the algorithm, the analysis goes through virtually unchanged, with the obvious change of now using $m=n \binom{\ell + k-2}{k-2}$ rather than $m=n^k$ when applying Claim~\ref{lem:granularity}.
%In place of \Cref{hehebound} we now get that the sample complexity is \dots} \rnote{Finish this - I think we will use the new estimation procedure that someone seems to be putting into Section 6 even as I type this}

%}

\begin{flushleft}
\bibliography{allrefs}{}
\bibliographystyle{alpha}
\end{flushleft}

\appendix

\end{document}